\documentclass[review]{elsarticle}

\usepackage{lipsum}
\usepackage{mathtools}
\usepackage{cuted}
\usepackage{amsmath}
\usepackage{bm}
\usepackage{amssymb} 
\usepackage{epsfig} 
\usepackage{subfigure}
\usepackage{textcomp}
\usepackage[utf8]{inputenc}
\usepackage[english]{babel}
\usepackage{lipsum}
\usepackage{mathtools}
\usepackage{cuted}
\usepackage{amsthm}
\newtheorem{definition}{Definition}

\journal{Springer Nature}

\begin{document}

\begin{frontmatter}

\title{On minimum phase transformation and filter design}

\author{J.C. Olivier and E. Barnard}
\address{J.C. Olivier is with the School of Engineering, University of Tasmania, Sandy Bay Road, Hobart, Australia. jc.olivier@utas.edu.au\\
E. Barnard is with Multilingual Speech Technologies (MuST), North-West University, South Africa.  }




\begin{abstract}
Minimum-phase finite impulse response filters are widely used in practice, and much research has been devoted to the design of such filters. However, for the important case of Chebyshev filters there is a curious mismatch between current best practice and well-established theoretical principles. The paper shows that this difference can be understood through analysis of the time-domain factorization of a suitable extended matrix. This analysis explains why the definition of  a \emph{factorable} linear phase filter must be revised.  The time domain analysis of factorization suggests initial values leading to fast and accurate convergence of iterative algorithms for the design of minimum-phase finite impulse response filters. Numerical results are provided to demonstrate that a significant improvement in filter tap accuracy is obtained when the well-established theoretical principles are correctly applied to the design of minimum phase finite impulse response filters.
\end{abstract}

\begin{keyword}
Minimum phase transformations, FIR filter design, factorization
\end{keyword}

\end{frontmatter}

\section*{Statement and Declarations}

\begin{itemize}
    \item The manuscript has no associated data. 
    \item There are no financial or non-financial interests that are directly or indirectly related to the work submitted for publication.
\end{itemize}

\section{Introduction}

Linear filter theory plays an important role in systems analysis, feedback control,  signal processing, music, acoustics, communications theory, and is well understood \cite{oppenheim,oppenheim2}.  A  linear filter that is widely applied in practice  is the discrete-time finite impulse response (FIR) filter \cite{smith,proakis2}. 

Depending on the intended application of the FIR filter, a linear phase filter is often required, as linear phase guarantees that the system will not introduce distortion \cite{linear_syst}.  The design of linear phase FIR filters is mature and software is available for the efficient design of such filters.    For the specific case of a linear phase FIR filter with a Chebyshev spectral approximation,  an optimal design  was proposed in \cite{parks}\footnote{The so-called Parks-McClellan design is available in MATLAB as the function \emph{firpm}.}.  For the same filter order, the stopband attenuation achieved by a Chebyshev approximation exceeds that of a Butterworth approximation, and Chebyshev filters can achieve a sharper transition between the passband and the stopband.  

 However under certain conditions a  minimum phase response may offer advantages to the system designer:  
 \begin{enumerate}
 \item Minimum phase FIR filters  are unconditionally stable when feedback is applied.
 \item  They have an optimal step response, which is desirable in feedback control systems.  
 \item They are robust when the FIR tap coefficients are discretized. 
 \item They  contain the least number of taps able to achieve a specified filter magnitude response.  
\end{enumerate}

 Several design methodologies for the design of a Chebyshev minimum phase filter have been proposed.   The literature can broadly be  categorized as promoting minimum phase FIR filter design based on the Hilbert transform \cite{evans}, explicit enumeration of the polynomial roots (root finding)  \cite{kamp,smith}, the complex cepstrum \cite{reddy} and spectral factorization \cite{herrmann}.   The Hilbert transform and the cepstrum are both based on the discrete Fourier transform (DFT)\footnote{Implemented through the fast Fourier transform (FFT). } and require very long FFT's to obtain good FIR filter performance.   As the root finding algorithm explicitly computes all the zeros on the $Z$ domain,  this approach is effective for low order systems,  but for high order systems becomes numerically unstable \cite{smith}.  

 Recently new results were published \cite{antonio}, demonstrating that Chebyshev filter design based on spectral ($Z$ domain) factorization \cite{herrmann,o-w} is computationally efficient, and yields numerical results superior to any other method currently available  in the literature.  
 
 This paper proposes new results for  the design of Chebyshev minimum phase filters, based on factorization of a linear phase filter $\mathbf g$. Orchard and Wilson showed \cite{o-w} that the taps of a  minimum phase filter must solve a certain system of non-linear equations.  The $\mathrm L_2$ norm of the residual error vector (for the system of non-linear equations) has been proposed as a metric to measure the quality of a minimum phase filter design \cite{antonio}.  This paper demonstrates  that an \emph{optimal} Chebyshev minimum phase filter has, at least theoretically, a residual error  norm that is zero.  It will be shown that such an optimal minimum phase filter requires the Gramian matrix representing the linear phase filter $\mathbf g$ to be positive definite. 
 
 A detailed analysis will be presented to show that the residual error for a design based on lifting  \cite{herrmann} is finite  --- the transfer functions of such filters are (theoretically) positive semi-definite and sub-optimal, since complex filter taps are required to solve the Orchard-Wilson equations exactly.  The paper demonstrates that real minimum phase filter taps and zero residual error are possible, if and only if, the Gramian matrix representing $\mathbf g$ is positive definite. On a digital computer with a finite resolution the norm of the residual error is of course finite, but limited only by the resolution of the machine; that is, a computer with infinite precision would produce a residual error norm of zero.  This paper demonstrates that the residual error norm based on a $64$ bit MATLAB implementation is orders of magnitude smaller than that reported in \cite{antonio}, provided the filter $\mathbf g$ is positive definite.  
 
  Results are presented to study the effect of lifting the spectral domain response as proposed in \cite{herrmann}.  Denoting the lifting factor as $\gamma$,  the paper demonstrates that the residual error norm as a function of $\gamma$ exhibits a \emph{waterfall point}, beyond which the error norm \emph{falls away to zero} (at least theoretically).  The waterfall point coincides with the value of  $\gamma$ that renders the linear phase filter $\mathbf g$ positive definite, and this result is  proved based on factorization in the time domain. It is proved that regularization of a Gramian matrix yields the correctly adjusted linear phase filter $\mathbf g$, and that factorization follows only if the Gramian is positive definite.   It is shown that the magnitude of the smallest eigenvalue of the Gramian matrix is identical to the lifting value proposed in \cite{herrmann}\footnote{Note that the lifting value must be based on a measurement of the \emph{realized} filter ripple, not the design values.}. Section \ref{min_phase} will provide numerical  results based on complex analysis (obtained with Mathematica) to demonstrate these results. 
 
 The paper also considers a second application, namely the transformation of a given arbitrary phase FIR to a minimum phase FIR filter with an identical spectral magnitude.  There are several applications in practice that will require such a transformation. One such case is where a medium or channel is  characterized by an impulse response \cite{linear_syst} containing a random phase \cite{smith,proakis1}.  Another possibility is that only the magnitude of the frequency response of a propagation medium is known, but a minimum phase response is required for computational reasons \cite{optics}.   Under these conditions the transformation of a given FIR to a minimum phase FIR is required. 
 
An approach often deployed to perform this transformation is based on estimating the coefficients of the minimum phase filter, and of all the estimators available the minimum mean square error (MMSE) estimator is most often deployed \cite{cioffi}.  This paper presents  numerical results to demonstrate that factorization provides an efficient solution for this transformation, and yields results that significantly outperform the MMSE method.

The paper is structured as follows.  Section \ref{min_phase} reviews the theory of minimum phase FIR design through factorization on the spectral (Z) domain. This section also serves to make the paper somewhat self contained, and presents a detailed and critical analysis of the requirements for factorization.   In Section \ref{time_domain} factorization in the discrete-time domain is analyzed, and it is shown that the smallest (and also negative) eigenvalue of the Gramian matrix is equivalent to the peak negative value of the  amplitude frequency response of $\mathbf g$. This section demonstrates that the Gramian matrix representing the linear phase filter $\mathbf g$ must be positive definite to make factorization possible.    Section \ref{FIR_design} presents the proposed optimal design of a positive definite Chebyshev minimum phase FIR filter, and demonstrates that numerical results significantly outperform a design based on a positive semi-definite linear phase filter, as well as best practice design available in the literature  \cite{antonio}. Section \ref{transform} presents numerical results for the transformation of a given arbitrary phase FIR to a minimum phase FIR. The numerical results are compared to results obtained through MMSE design.  The  paper is concluded in Section \ref{conclude}. 

\section{A critical review of factorization on the spectral domain} \label{min_phase}


\subsection{FIR filters, minimum phase and frequency response}

Denote the finite impulse response (FIR) of a linear system as a column vector $\mathbf h = \{ h[0],h[1],\cdots,h[M-1] \}^{\mathrm T}$, with $M$ elements or taps. The symbol $^{\mathrm T}$ denotes the transpose operation.  To compute (measure) the FIR,  the input of the  system is set to the Kronecker delta $\delta[n]$, and the computed (measured) output is by definition the FIR $\mathbf h$.  Here $n$ indicates discrete time, and a causal system has $h[n] = 0 ~\forall ~ n < 0 $.  The convolution theorem for causal systems states that \cite{oppenheim} 
\begin{equation}
y[n] = \sum_{k=0}^{M-1} h[k] x[n-k]
\end{equation}
where $x[n]$ denotes the system input at time $n$, and $y[n]$ denotes the system output at time $n$.  

The discrete-time FIR $\mathbf h$ is  related to the transfer function denoted $H(z)$, based on the $Z$ transformation of $\mathbf h$ \cite{linear_syst}.  The unilateral $Z$ transform is applicable to causal systems, and  transforms the discrete-time FIR $\mathbf h$ to a complex spectral-domain ($Z$ domain) representation,  given by  
\begin{equation}
H(z) = Z \{ h[n] \}  = \sum_{k=0}^{\infty} h[n] z^{-n}.
\end{equation}

The power spectral magnitude of the FIR filter  is defined as $|H(\Omega)|^2$  where 
\begin{equation}
H(\Omega) = H(z = e^{i \Omega}).
\end{equation}
$H(\Omega)$ is referred to as the frequency response, with $-{\pi} \le \Omega \le {\pi}$ the normalized frequency.       When the transfer function is rational and given by 
\begin{equation}
H(z) = \frac{P(z)}{D(z)},
\end{equation}
then in general there are $N$ \emph{poles}, defined as the set of all samples $z_j$ where $D(z_j) = 0, ~\forall~ j \in \{ 1,2,3,\cdots,N \}$.  There are also $M$ \emph{zeros}, defined as the set of all samples $z_j$ where $P(z_j) = 0, ~\forall~ j \in \{ 1,2,3,\cdots,M \}$.  It can be shown \cite{linear_syst} that a stable linear system characterized by a FIR has all its poles at $z=0$, and has in general a finite number of zeros in the complex $Z$ domain (related to the order of $\mathbf h$, that is, the number of taps in $\mathbf h$).   

It is possible to transform $\mathbf h$ so that only the phase of the complex frequency response $H(\Omega)$ is modified.  For such a transformed system denoted as $C(z)$, the magnitude of the transformed frequency response is    $|C(\Omega)| = |H(\Omega)|$.  Let there be $S$ possible transforms that will satisfy  the requirement $|C(\Omega)| = |H(\Omega)|$, then there are $S$  FIR vectors  $\mathbf c_{q \in \{1,\cdots,S\}}$ with identical magnitude spectra.  However there is a unique minimum phase FIR, denoted $\mathbf c_p$, which has the fastest  decay in the time domain --- in the sense that for any $0 \le k \le M-1$  \cite{smith}
\begin{equation}
\sum_{n=0}^{k} \, |{c_p[n]}|^2  >  \sum_{n=0}^{k} \, |{c_q[n]}|^2 ~\forall~ q \neq p.
\end{equation}

Viewed on the $Z$ domain, $C_p(z)$  follows through the $Z$ transformation of $\mathbf c_p$, and $C_p(z)$ does not have poles or zeros outside the unit circle.

\subsection{Factorization of a linear phase filter on the $Z$ domain}

\subsubsection{A  factorable linear phase filter $\mathbf g$}

\begin{figure} []
\centering
  \includegraphics[width=0.7\textwidth]{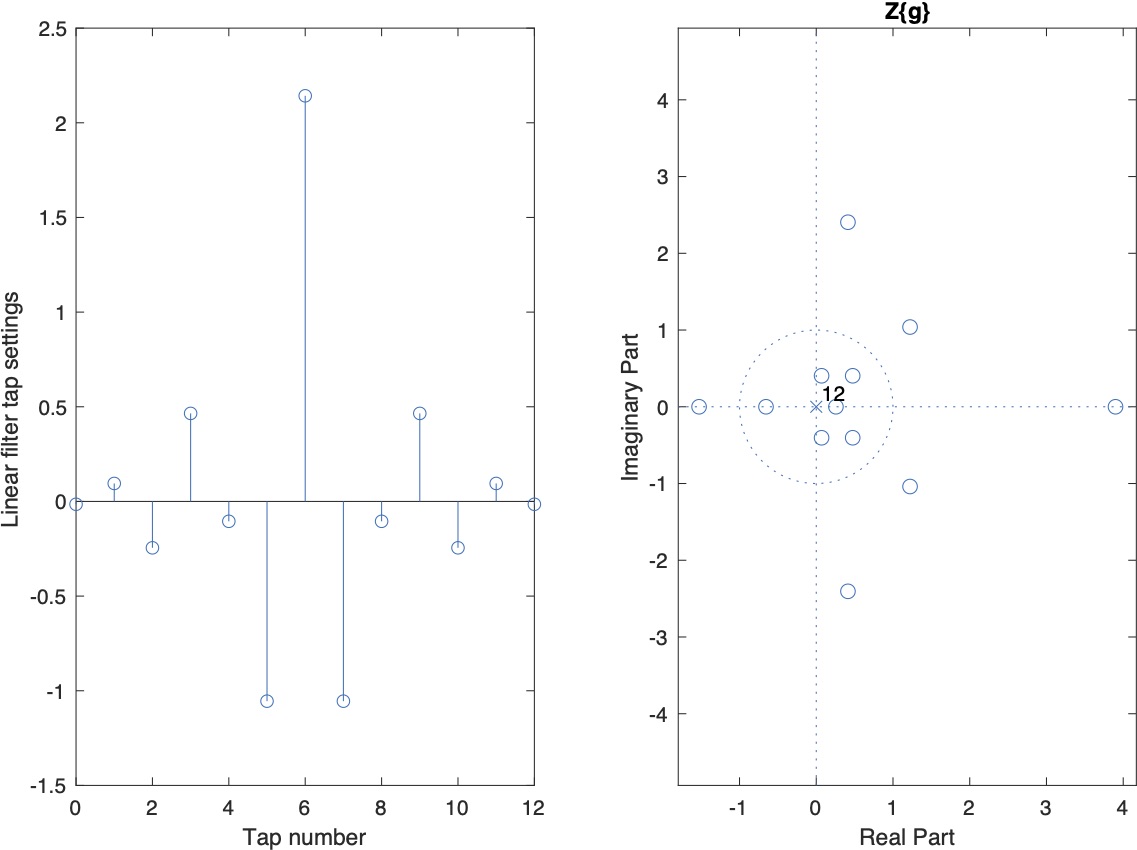}
\caption{The time and $Z$ domain representation of the linear phase filter $G(z)$. }
\label{figure1}    
\end{figure}
Consider a linear phase filter $\mathbf g$ with $13$ taps, as shown in Figure \ref{figure1}.  All the zeros on the $Z$ domain occur in pairs, such that for any zero $z_i$ on the $Z$ domain, there is also a zero $\frac{1}{z_i}$ present. Note there are no zeros located on the unit circle in this case.  

 Here all zeros occur as inverse pairs,  and the literature denotes $G(z)$ as  \emph{factorable} if 
  \begin{equation} 
G(z) =  H(z) \, \hat{H}(z),
\end{equation}
where $H(z)$ represents the $Z$ domain representation of a filter $\mathbf h$ (with zeros $z_i$), and $\hat{H}(z)$ represents the $Z$ domain representation of the time reflected filter $\mathbf{h}_{\mathrm{rfl}} $ (with zeros $\frac{1}{z_i}$) \cite{linear_syst}.

There are several ways to assign zeros to the filters  $H(z)$ and  $\hat{H}(z)$, but if all the zeros inside the unit circle are assigned to  $C(z)$, then all the zeros outside  the unit circle  are assigned to  $\hat{C}(z)$.  In this case it follows that 
 \begin{equation}  \label{ow1}
G(z) =  C(z) \, \hat{C}(z) 
\end{equation}
and  the discrete-time filter $\mathbf c$ is minimum phase.  The filter $\mathbf c$ is depicted in Figure \ref{figure2} along with its $Z$ domain representation $C(z)$, where all the zeros are inside the unit circle. 

\begin{figure} [b]
\centering
  \includegraphics[width=1\textwidth]{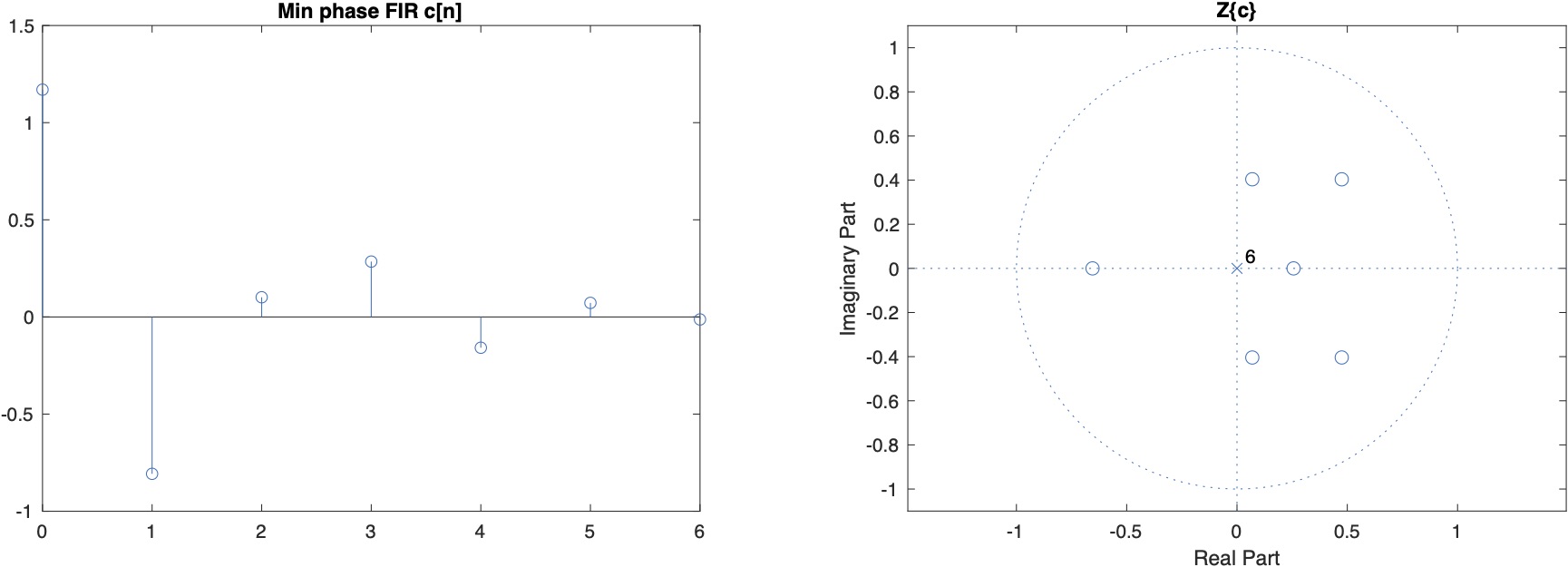}
\caption{The  minimum phase filter $\mathbf c$ in the time and $Z$ domain.}
\label{figure2}    
\end{figure}
Based on the $Z$ transform $C(z)$ can be written as 
\begin{equation}  \label{ow2}
C(z) = c[0] + c[1] z^{-1} + \cdots + c[M-1] z^{-(M-1)}.
\end{equation}
Hence the maximum phase filter is given by 
\begin{equation}  \label{ow3}
\hat{C}(z) = c[M-1] + c[M-2] z^{-1} + \cdots + c[0] z^{-(M-1)}.
\end{equation}
Combining (\ref{ow1}), (\ref{ow2}) and (\ref{ow3})
yields a system of non-linear equations to be solved to find the coefficients $\mathbf c$, given by \cite{o-w}
\begin{eqnarray} \label{nonlinear_system2}
\nonumber c_0^2 + c_1^2 + \cdots + c_{M-1}^2 &=& g\left [{M-1}\right ] \\ 
\nonumber c_0 \,c_1 + c_1 \, c_2 + \cdots + c_{M-2} \, c_{M-1} &=&  g\left [{M-2}\right ] \\
c_0 \, c_2 + c_1 \, c_3 + \cdots + c_{M-3} \, c_{M-1} &=&  g\left [{M-3}\right ] \\
\nonumber 
\vdots  &=&  \vdots \\
\nonumber c_0 \, c_{M-1} &=& g_{\mathrm{}}[{0}].
\end{eqnarray}
Hence it is evident that any linear phase filter $\mathbf g$ that is factorable will yield a minimum phase filter $\mathbf c$ through the solution of the Orchard and Wilson nonlinear equations (\ref{nonlinear_system2}).  The residual error vector $\mathbf e = \{ e_0,e_1,\cdots \}^T$ is defined as
\begin{eqnarray}
\nonumber c_0^2 + c_1^2 + \cdots + c_{M-1}^2 - g\left [{M-1}\right ]  &=& e_0\\ 
\nonumber c_0 \,c_1 + c_1 \, c_2 + \cdots + c_{M-2} \, c_{M-1} -  g\left [{M-2}\right ] &=& e_1 \\
c_0 \, c_2 + c_1 \, c_3 + \cdots + c_{M-3} \, c_{M-1} -  g\left [{M-3}\right ] &=& e_2\\
\nonumber 
\vdots  &=&  \vdots \\
\nonumber c_0 \, c_{M-1} -  g_{\mathrm{}}[{0}] &=& e_{M-1}.
\end{eqnarray}

The norm of the residual error is given by $\mathrm E_{L_2} = \sqrt{\mathbf e^T \mathbf e}$, and the  literature calls for the residual error to be deployed as a metric to measure the quality of a minimum phase filter design \cite{antonio}. In the next subsection, it will be shown that the residual error is more than just a metric to measure the quality of a minimum phase filter design --- in fact it plays a central role in this paper through a revised definition of what a factorable linear phase filter is.   

\subsubsection{A revised definition: a factorable linear phase filter} \label{factorable_def}

Some linear phase filters deploy zeros on the unit circle.  This makes it possible to generate an equiripple in the stopband, and is the basis of Chebyshev filter design.  An optimal linear phase Chebyshev filter design was proposed by Parks and McClellan \cite{parks}.  

However, these filters are typically not factorable. To demonstrate why this is the case, consider a linear phase Chebyshev filter with $5$ taps as indicated in Table \ref{voorbeeld}. The magnitude of the frequency response and the $Z$ domain representation of the filter $\mathbf g$ are shown in Figure \ref{g_studied}. 
  \begin{figure} []
\centering
  \includegraphics[width=0.65\textwidth]{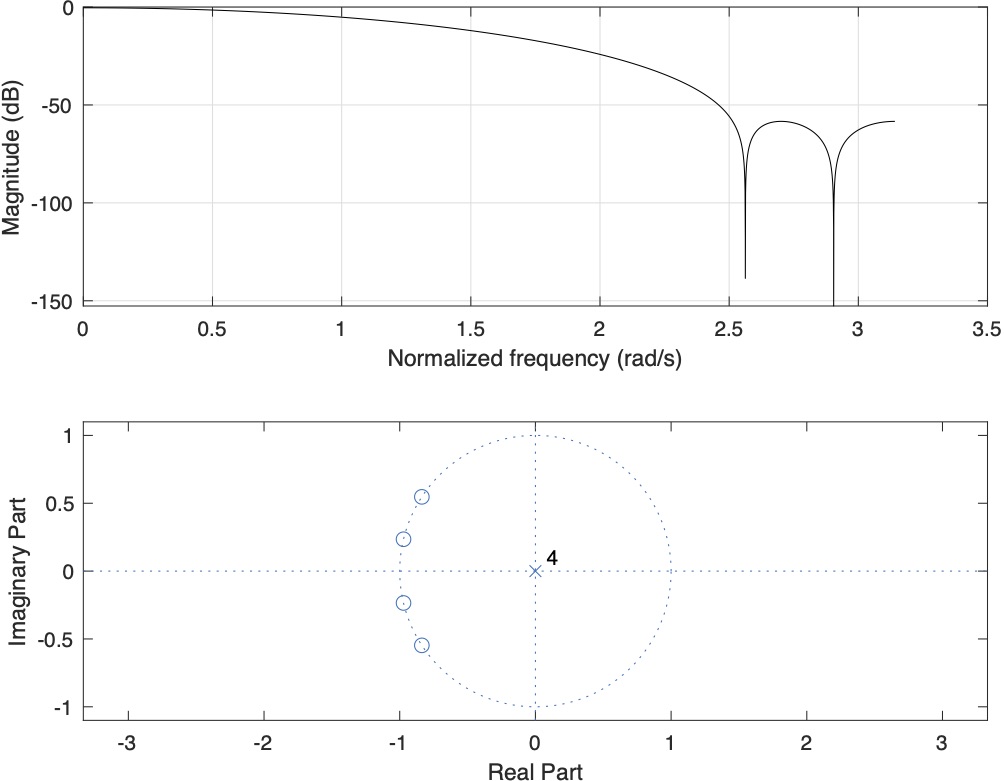}
\caption{The  linear phase filter $\mathbf g$.}
\label{g_studied}    
\end{figure}

From the $Z$ domain representation its clear that all the zeros are located on the unit circle in this case, and the zeros do not occur as inverse pairs.  Thus the factorization given by  (\ref{ow1}) is not possible in this case. This becomes clear when  the non-linear equations (\ref{nonlinear_system2}) are solved analytically, as all the solutions provided by Mathematica are complex.  Thus any solution of (\ref{nonlinear_system2}) where  $\mathbf c$ is required to have real taps, will inevitably yield a solution with $\mathrm E_{L_2} > 0$. 

An important step towards solving this problem was taken by Herrmann \emph{et al}  \cite{herrmann}, who proposed to add a constant $\gamma$ to the dominant tap of $\mathbf g$, which yields an adjusted filter $\mathbf{g_\mathrm{adj}}$, and is known as lifting. The idea is to prevent the magnitude from becoming negative at any frequency.  In this case there is only one negative stopband ripple\footnote{For positive $\Omega$.}, which is shown in Figure \ref{gap_explained} with $\gamma$ indicated.  The adjusted filter $\mathbf g_{\mathrm{adj}}$ is also shown, where clearly the magnitude is now no longer negative, and  the linear phase FIR filter $\mathbf{g_\mathrm{adj}}$ is referred to as positive.  In the next section it will be shown that the Gramian matrix corresponding to $\mathbf{g_\mathrm{adj}}$ has a smallest eigenvalue that is 0 --- in the language of linear algebra it is positive semi-definite. Thus the linear phase filter obtained through lifting as proposed in \cite{herrmann} is referred to as positive semi-definite in this paper.  

\begin{table}[]
\scriptsize
\caption{A low order linear phase filter $\mathbf g$} 
\centering 
\begin{tabular}{c c c c} 
\hline\hline 
Tap number $n$ & $\mathbf g$  \\ [0.5ex] 
\hline 
0 &  0.066075742625345 \\
1 &   0.239064282650394 \\
2 &  0.347182106755652 \\
3 &  0.239064282650394  \\ 
4 &   0.066075742625345 \\[1ex] 
\hline 
\end{tabular}
\label{voorbeeld} 
\end{table}

For the filter taps shown in Table \ref{voorbeeld} the value of $\gamma$ based on lifting \cite{herrmann} is given by 
\begin{equation}
    \gamma = {0.00120505352635236},
\end{equation}
and the zeros have migrated to form double zeros on the $Z$ plane as shown in Figure \ref{gap_explained}. The adjusted filter $\mathbf g_{\mathrm{adj}}$ now \emph{appears} to be factorable, as one of the zeros can be associated with the causal minimum phase filter $\mathbf c$, and the other zero with the anti-causal reflected version of $\mathbf c$. This was indeed the argument presented by Herrmann, and deployed by the authors in \cite{o-w,antonio}.  
 \begin{figure} []
\centering
  \includegraphics[width=0.8\textwidth]{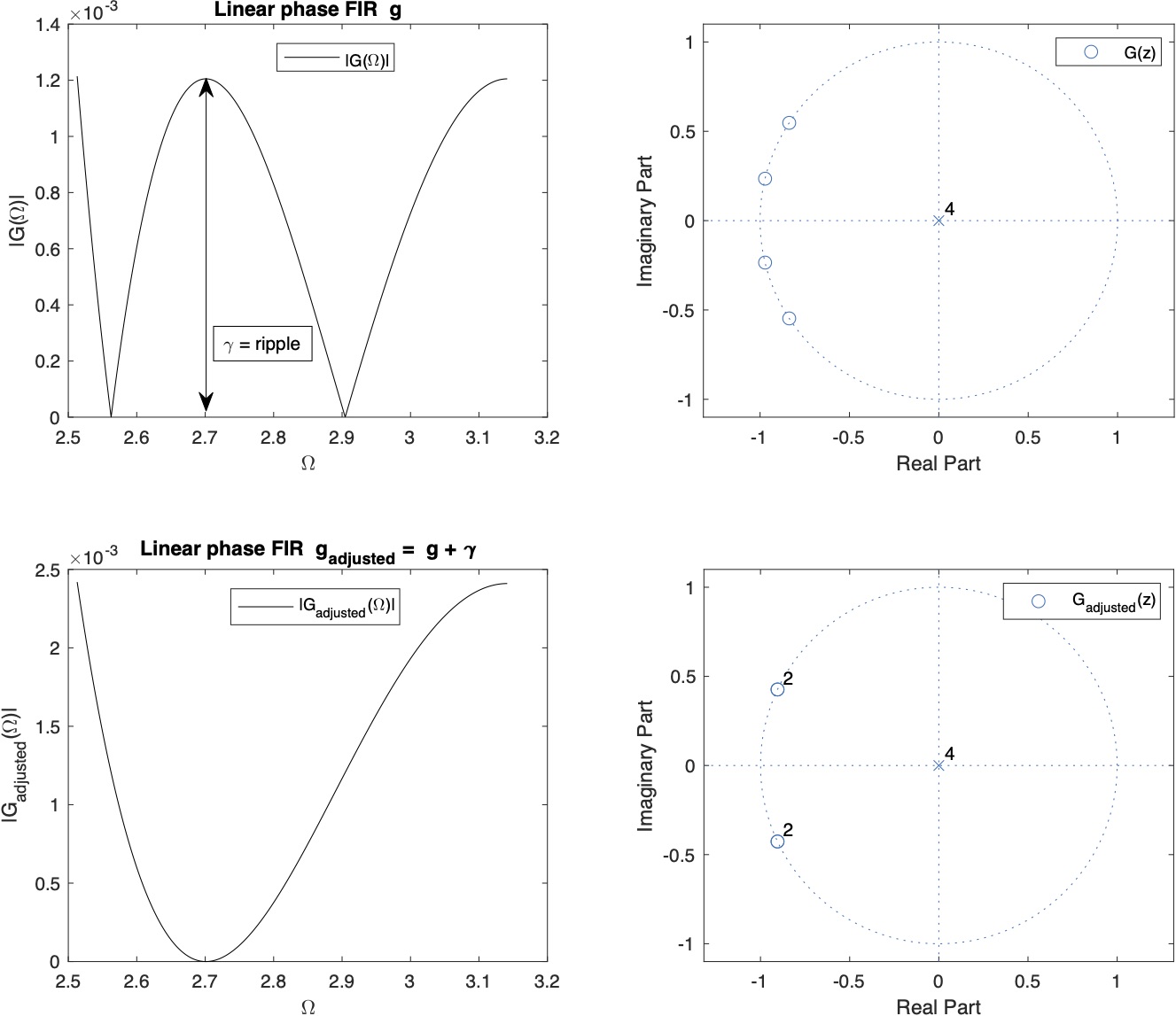}
\caption{Lifting as proposed by Herrmann \cite{herrmann} yields an adjusted filter denoted $\mathbf g_{\mathrm{adj}}$.}
\label{gap_explained}    
\end{figure}

However, consider an analytical solution in the complex plane of the nonlinear equations  (\ref{nonlinear_system2}) for $\mathbf{g_\mathrm{adj}}$. Such a solution can be obtained by making use of Mathematica, the outcome of which is shown in Figure \ref{mathematica2}.  It is clear that the  positive semi-definite linear phase filter $\mathbf{g_\mathrm{adj}}$ does not yield a real solution for the system of non-linear equations given by (\ref{nonlinear_system2}).  All the solutions provided by Mathematica are complex.  Thus any attempt to solve (\ref{nonlinear_system2}) with the taps of the minimum phase filter $\mathbf c$ real (as they should be), will naturally lead to $\mathrm E_{L_2} > 0$. Hence lifting as proposed by Herrmann \cite{herrmann} does not provide a factorable linear phase filter, if \emph{factorable} is defined by the requirement to have the Orchard and Wilson equations (\ref{nonlinear_system2}) yield a real solution with $\mathrm E_{L_2} = 0$. This observation leads to a revised   definition of \emph{factorable} that will be adopted in this paper: 

\begin{definition} \label{factorable}
A linear phase filter $\mathbf g_{\mathrm aug}$ is factorable, if and only if, $\mathrm E_{L_2} = 0$ and the factor $\mathbf c$ is real.
\end{definition}

It will now be demonstrated how lifting can be deployed to yield a factorable linear phase filter. Consider the effect of adding a  small number denoted by $\epsilon$ to $\gamma$, say 
\begin{equation}
    \gamma^\star  = \gamma + 3 \times 10^{-17}.
\end{equation}
Mathematica yields a real solution if $\gamma^\star$ is deployed, as shown in Figure \ref{mathematica1}.  It can be verified that for $\gamma^\star$  it follows that $\mathrm E_{L_2} = 0$, and thus $\gamma^\star$ provides a factorable linear phase filter. In the next section it will be shown that the linear phase filter $\mathbf g_{\mathrm{aug}}$ based on $\gamma^\star$ yields a Gramian matrix that is positive definite. 

Results based on the exact solution are of course  theoretical, and if a numerical solution is performed on a  digital computer with finite resolution,  then $\mathrm E_{L_2} > 0$.  But what should be realised is that this error is induced by finite machine resolution, and can be reduced simply by improving the computer resolution. 

Section  \ref{FIR_design} will provide numerical results to verify the definition above for a practical $25$ tap filter, and compare the results to best practice methodology available in the literature.

\begin{figure} []
\centering
  \includegraphics[width=0.5\textwidth]{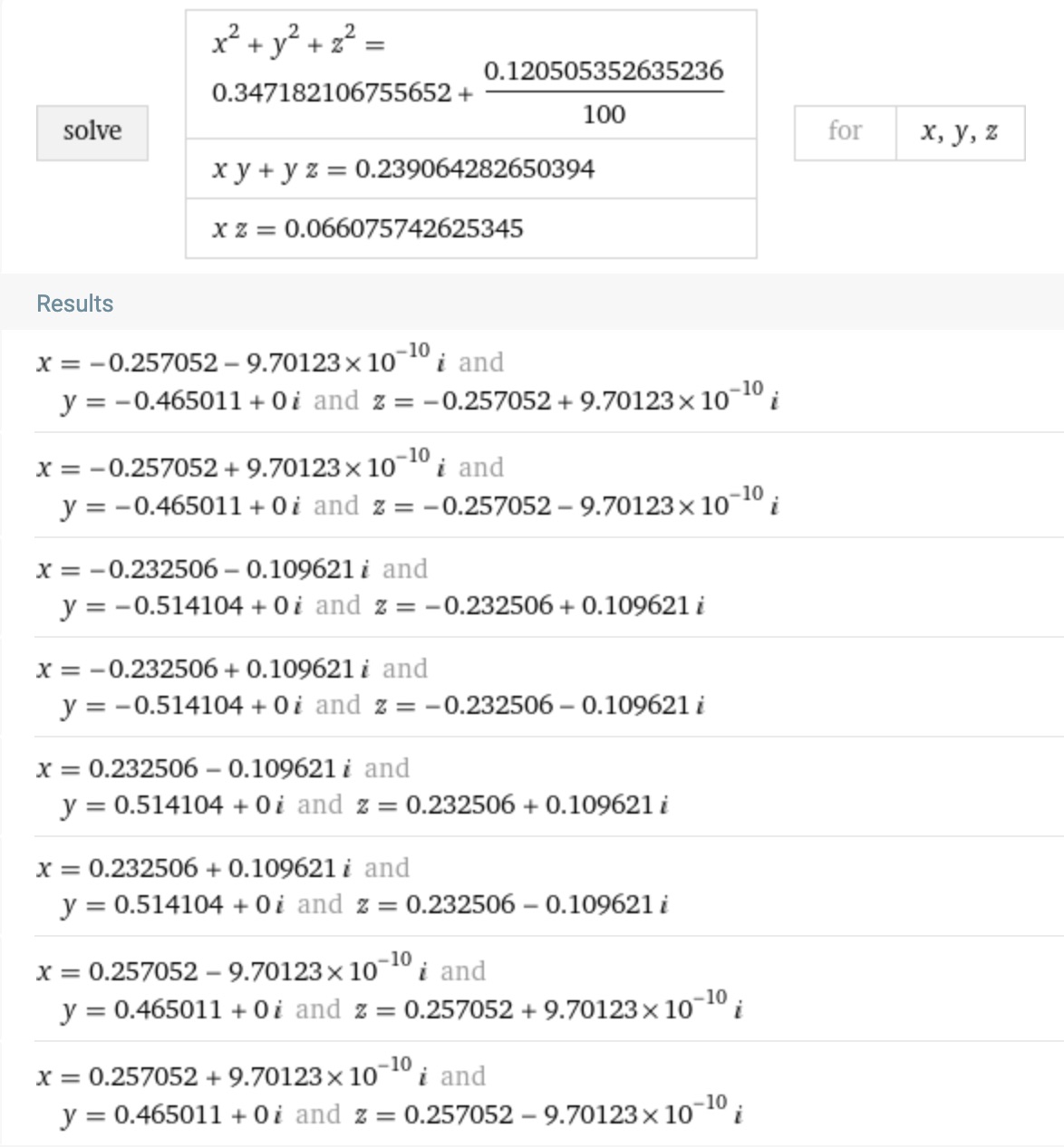}
\caption{The solution of the non-linear Orchard and Wilson equations for a positive semi-definite $\mathbf{g}_{\mathrm{adj}}$ based on Mathematica.}
\label{mathematica2}    
\end{figure}
\begin{figure} []
\centering
  \includegraphics[width=0.5\textwidth]{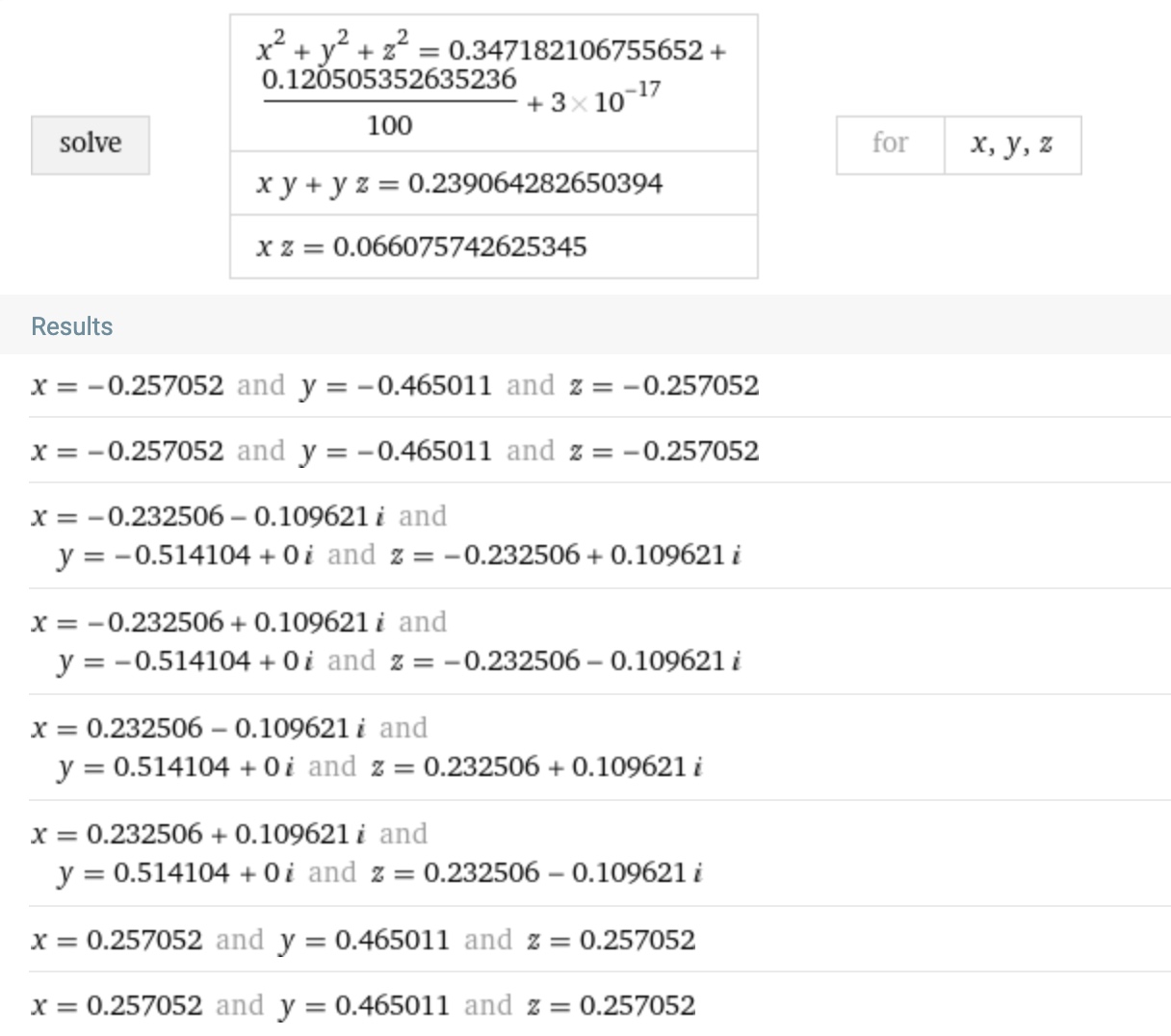}
\caption{The solution of the non-linear Orchard and Wilson equations for a positive definite $\mathbf{g}_{\mathrm{adj}}$ based on Mathematica.}
\label{mathematica1}    
\end{figure}


\section{Time domain factorization and minimum phase FIR filter design} \label{time_domain}

In this section minimum phase filter design based on time domain factorization  is presented, and it will be shown that factorization on the time domain confirms the revised definition  provided in the previous section.

\subsection{Locally Toeplitz matrices exhibiting  symmetry point equilibrium }

In anticipation that Cholesky decomposition will be required to factorize on the time domain \cite{cioffi}, consider the transpose of an upper triangular matrix  obtained through Cholesky decomposition, shown in Figure \ref{chol}. The matrix is not strictly Toeplitz, as the top rows and the bottom rows differ from the rows near the symmetry point, where it is \emph{locally Toeplitz}.   This is typical  when the matrix size is much greater than the number of diagonals that contain non-zero values. Such a matrix is exhibiting {symmetry point equilibrium}.  
\begin{figure} [h]
\centering
  \includegraphics[width=0.45\textwidth]{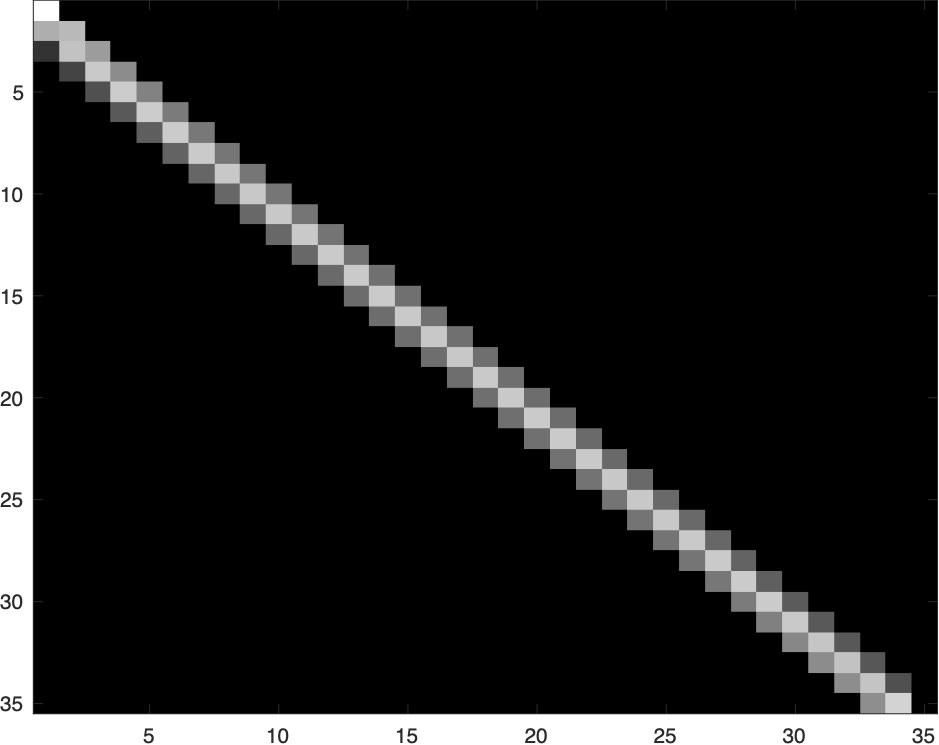}
\caption{The structure of a matrix where symmetry point equilibrium holds --- it is locally Toeplitz,  but not strictly Toeplitz.}
\label{chol}    
\end{figure}

 \begin{figure} []
\centering
  \includegraphics[width=0.8\textwidth]{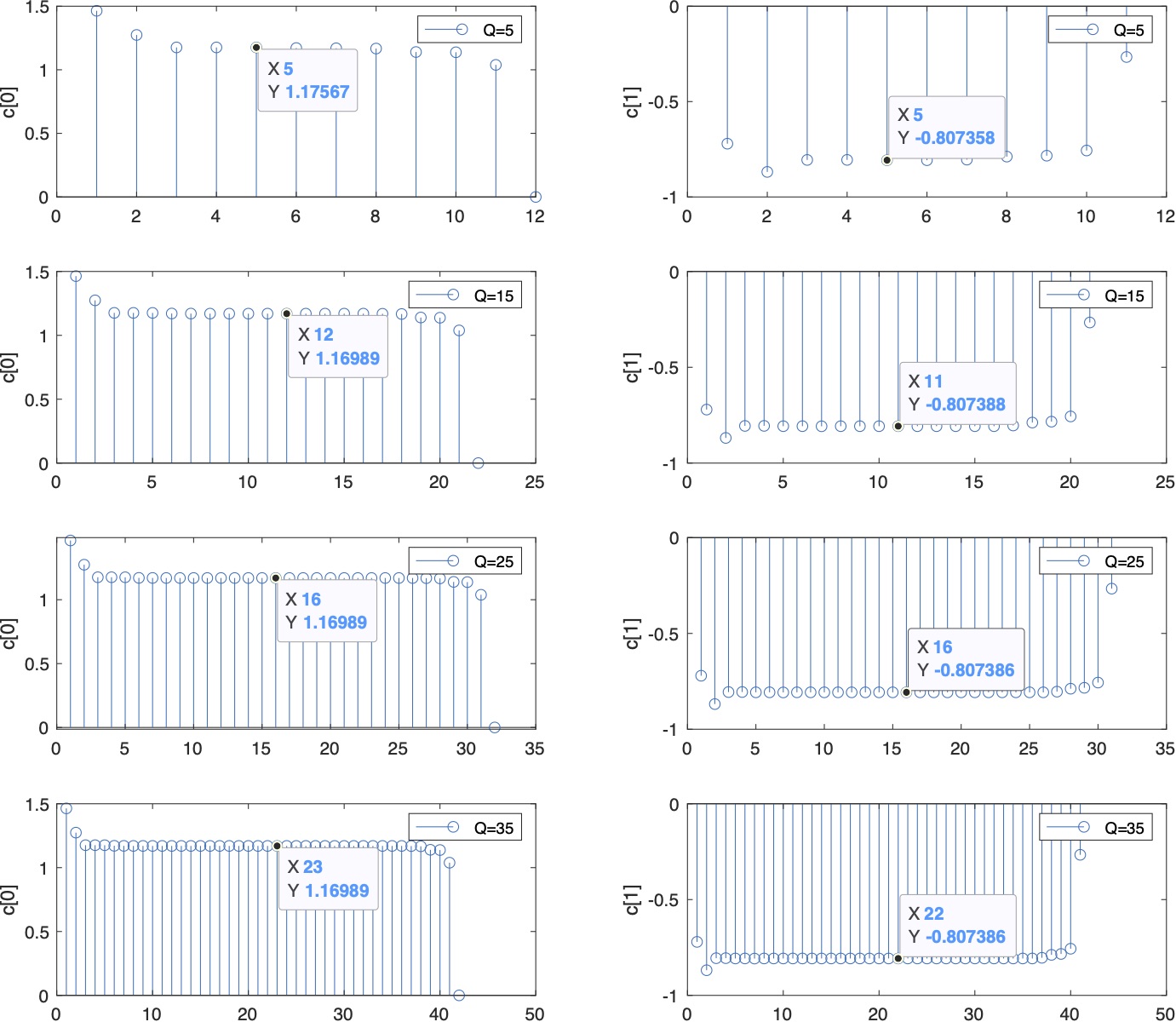}
\caption{A demonstration of symmetry point equilibrium as a function of $Q$.  }
\label{figure6}    
\end{figure}

 To demonstrate how such a matrix becomes locally Toeplitz as its size $Q$ is increased (while keeping the number of non-zero diagonals fixed),  Figure \ref{figure6} shows the main diagonal as well as the next diagonal of a matrix $\mathbf C$ obtained through Cholesky decomposition as a function of $Q$.  It is clear that as $Q$ becomes large,  rows near the symmetry point move toward an equilibrium state, with tap values converging to a fixed value and the matrix is locally Toeplitz.

A matrix that is locally Toeplitz can represent a \emph{time-invariant filter} provided it is operating on a vector that has \emph{only local support} near the symmetry point.   To make this statement clear, define a time invariant augmented FIR filter so that it has only local support, denoted as  $\mathbf h_\mathrm{aug}$ and defined as\footnote{The fact that the vector has time $n$ advancing from right to left is in anticipation of the fact that Cholesky factorisation implemented by most computing platforms provides a factorisation in the form $\mathbf U^T \mathbf U$ where $\mathbf  U$ is upper triangular}
\begin{equation}
\mathbf h_\mathrm{aug} = \{\underbrace{0,\cdots,0}_{\text{$Q$ zeros}},  h_{M-1},\cdots,h_1,h_0,\underbrace{0,\cdots,0}_{\text{$Q$ zeros}} \}^{\mathrm T}
\end{equation}
where $Q \gg M$.  As $\mathbf h_\mathrm{aug}$ has support only near the symmetry point,  a locally Toeplitz matrix where symmetry point equilibrium holds, performs a time-invariant convolution operation on $\mathbf h_\mathrm{aug}$,  even if the matrix is not strictly Toeplitz.

 \subsection{An augmented impulse response and the Gramian}
  
Define a Toeplitz matrix with $\left ( 2Q+M \right )$ columns and rows, given by 
\begin{eqnarray} \small  
\mathbf H = \left[ \begin{array}{lllllllllll}
h_0  & h_1 & \cdots & h_{M-1} & 0 & 0 & \cdots  & 0 \\
0 & h_0 & h_1 & \cdots & h_{M-1} & 0 & \cdots  & 0 \\
\vdots & \vdots  & \vdots & \vdots &  \vdots  & \vdots  & \vdots & \vdots \\
0 & \cdots & 0  & h_0 & h_1 & \cdots & \cdots  & 0 \\
0 & 0 & \cdots & 0  & h_0 & h_1 & \cdots   & 0 \\
\vdots & \vdots  & \vdots & \vdots &  \vdots  & \vdots  & \vdots & \vdots \\
0  & \cdots & \cdots & \cdots & \cdots & \cdots & 0  & h_0 \\
\end{array} \right].
 \end{eqnarray}
Now define  an augmented Kronecker delta as 
\begin{equation}
\mathbf \delta_\mathrm{aug} = \{\underbrace{0,\cdots,0}_{\text{$Q$ zeros}}, {\underbrace{0,\cdots,0}_{\text{${M-1}$ zeros}}},1,\underbrace{0,\cdots,0}_{\text{$Q$ zeros}} \}^{\mathrm T}
\end{equation}
then  it follows that 
\begin{equation} \label{augment_verg}
\mathbf h_{\mathrm{aug}} = \mathbf H~ \mathbf \delta_\mathrm{aug}.
\end{equation}
Define the Gramian matrix $\mathbf G$  as 
\begin{equation}
\mathbf G = \mathbf H^{\mathrm T} \, \mathbf H
\end{equation}
then it follows that 
\begin{equation}
\mathbf H^{\mathrm T} ~\mathbf h_{\mathrm{aug}} = \mathbf H^{\mathrm T} \, \mathbf H~ \mathbf \delta_\mathrm{aug} = \mathbf G ~\delta_\mathrm{aug}.
\end{equation}

$\mathbf G$ is symmetric, Toeplitz and Hermitian, and represents a linear phase FIR filter denoted as $\mathbf g$ \cite{linear_syst}.  The matrix $\mathbf G$ has causal non-zero diagonals (the main diagonal and upper triangular part), as well as non-zero diagonals that are anti-causal (the lower triangular part).  The dominant (centre) tap of $\mathbf g$  is represented by the main diagonal of $\mathbf G$.   

\subsection{An approximate minimum phase FIR $\mathbf c_{\mathrm{approx}}$}

This subsection derives an approximate minimum phase FIR $\mathbf c_{\mathrm{approx}}$, based on Cholesky decomposition of the Gramian.  The approximate filter $\mathbf c_{\mathrm{approx}}$ will be deployed as an initial guess  when the non-linear equations are solved through numerical optimization. 

On the time domain the linear phase filter $\mathbf g$ will be factorable, if and only if  its matrix representation $\mathbf G$ is  positive definite. Then  Cholesky factorization  can be performed to yield an upper triangular matrix $\mathbf C$ as \cite{linear_alg}
\begin{equation}
\mathbf G = \mathbf C^{\mathrm T} \mathbf C
\end{equation}
and it follows that  
\begin{equation}
\mathbf H^{\mathrm T} \mathbf h_{\mathrm{aug}} = \mathbf C^{\mathrm T} \mathbf C~ \mathbf \delta_\mathrm{aug}  ~\implies~ 
 \underbrace{\left [ (\mathbf C^{\mathrm T})^{-1} \mathbf H^{\mathrm T} \right ]}_{\text{Matrix~}\mathbf F}    \mathbf h_{\mathrm{aug}} = \mathbf C~ \mathbf \delta_\mathrm{aug}.
\end{equation}
 The matrices $\mathbf C$ and $\mathbf F$ are locally Toeplitz (if $Q$ is sufficiently large), and  since $\mathbf h_{\mathrm{aug}}$ and $ \mathbf \delta_\mathrm{aug}$ have only local support,  these matrices represent {time-invariant} filters.

 Thus for a sufficiently large value of $Q$ it follows that   
\begin{equation} \label{transf}
\mathbf{c_\mathrm{aug}} = \mathbf C~ \mathbf \delta_\mathrm{aug} =  \mathbf F \, \mathbf h_{\mathrm{aug}}
\end{equation}
where $\mathbf{c_\mathrm{aug}}$ is given by 
\begin{equation}
\mathbf c_\mathrm{aug} = \{\underbrace{0,\cdots,0}_{\text{$Q$ zeros}},c_{M-1},\cdots, c_1,c_0,\underbrace{0,\cdots,0}_{\text{$Q$ zeros}} \}^{\mathrm T}.
\end{equation}

Cholesky decomposition expresses a Hermitian positive definite matrix as the product of a minimum phase matrix and its match \cite{linear_alg}, regardless of the value of $Q$. Thus when symmetry point equilibrium holds, it follows that a minimum phase FIR filter can be recovered from $\mathbf c_{\mathrm{aug}}$ as  
\begin{equation}
\mathbf c_{\mathrm{approx}} = \{ c_0,c_1,c_2,\cdots, c_{M-1} \}^{\mathrm T}.
\end{equation}

The FIR $\mathbf c_{\mathrm{approx}}$ is approximate as $Q$ is finite, but it will be shown in  Section \ref{FIR_design} that $\mathbf c_{\mathrm{approx}}$ is remarkably accurate, even for moderate settings of $Q$.  Section \ref{FIR_design}   will  also demonstrate that $\mathbf c_{\mathrm{approx}}$ is an appropriate choice as an initial guess to perform numerical optimization of  the non-linear equations given by (\ref{nonlinear_system2}).  

In Subsection \ref{formal} it will be formally shown that if $Q\rightarrow \infty$ then $\mathbf c_{\mathrm{approx}} \rightarrow  \mathbf c$.

\subsection{The optimal setting for $\gamma$ }

The matrix $\mathbf G$  has real eigenvalues \cite{linear_alg},  and  as remarked above will be factorable based on Cholesky factorization if $\mathbf G$ is positive definite.  Thus in order to be factorable, the minimum eigenvalue of matrix $\mathbf G$ must satisfy
\begin{equation}
\lambda_{\mathrm{min}} > 0.
\end{equation}
Depending on the application $\mathbf G$ may well be positive definite, but for Chebyshev filter design $\mathbf G$ is not positive  definite, and the minimum eigenvalue will be negative. Hence it follows that the Gramian will not be factorable through Cholesky factorization.  To mitigate this problem and to guarantee that the Gramian is factorable, the main diagonal can be modified, with the modified matrix denoted  as $\mathbf G_{\mathrm{adj}}$ and   given by 
\begin{equation}
\mathbf G_{\mathrm{adj}} = \mathbf G + \gamma \mathbf I.
\end{equation}
 $\mathbf G_{\mathrm{adj}}$ is guaranteed to be factorable as 
 \begin{equation}
 \mathbf G_{\mathrm{adj}} = \mathbf C^{\mathrm T} \mathbf C 
 \end{equation} 
 if, and only if, $\gamma$ is chosen as 
\begin{equation} \label{gamma_cond}
\gamma > |{\lambda_{\mathrm{min}}}| 
\end{equation}
as then the smallest eigenvalue is finite and positive. This result is known  as regularization \cite{linear_alg,regular}.  

To demonstrate that regularization in the discrete-time domain is equivalent to lifting  in the $\Omega$ domain, the filter shown in Table \ref{voorbeeld} is considered again as an example.   For this case the minimum eigenvalue of the Gramian $\mathbf G$ is negative and $\mathbf g$ is not factorable.  Figure \ref{eiewaarde} shows the value of the minimum eigenvalue of $\mathbf G$ as a function of $Q$, and it is clear that for a large value of $Q$, $ |\lambda_{\mathrm{min}}|$ is converging towards the value of the stopband ripple peak (negative) as shown in Figure \ref{gap_explained}.

On the discrete time domain it is a requirement for $\mathbf g_\mathrm{adj}$ to be positive definite, and thus it is required that $\gamma > |\lambda_{\mathrm{min}}|$.  In the next subsection it will be proved that if $Q \rightarrow \infty$, Cholesky factorization solves the Orchard and Wilson nonlinear equations (\ref{nonlinear_system2}). And since it is known that in the limit $Q \rightarrow \infty$ the smallest (negative) eigenvalue of $\mathbf G$ is equivalent to the stopband negative ripple peak, the result given by (\ref{gamma_cond}) suggests that a positive definite linear phase filter $\mathbf g_\mathrm{adj}$ must be deployed and shown to be factorable. That is,   determine $\gamma$ by solving the nonlinear equations through computer optimization so that $E_{\mathrm{L_2}}$ is limited only by the machine resolution.  This will be the basis of the filter design presented in Section \ref{FIR_design}.

\begin{figure} []
\centering
  \includegraphics[width=0.65\textwidth]{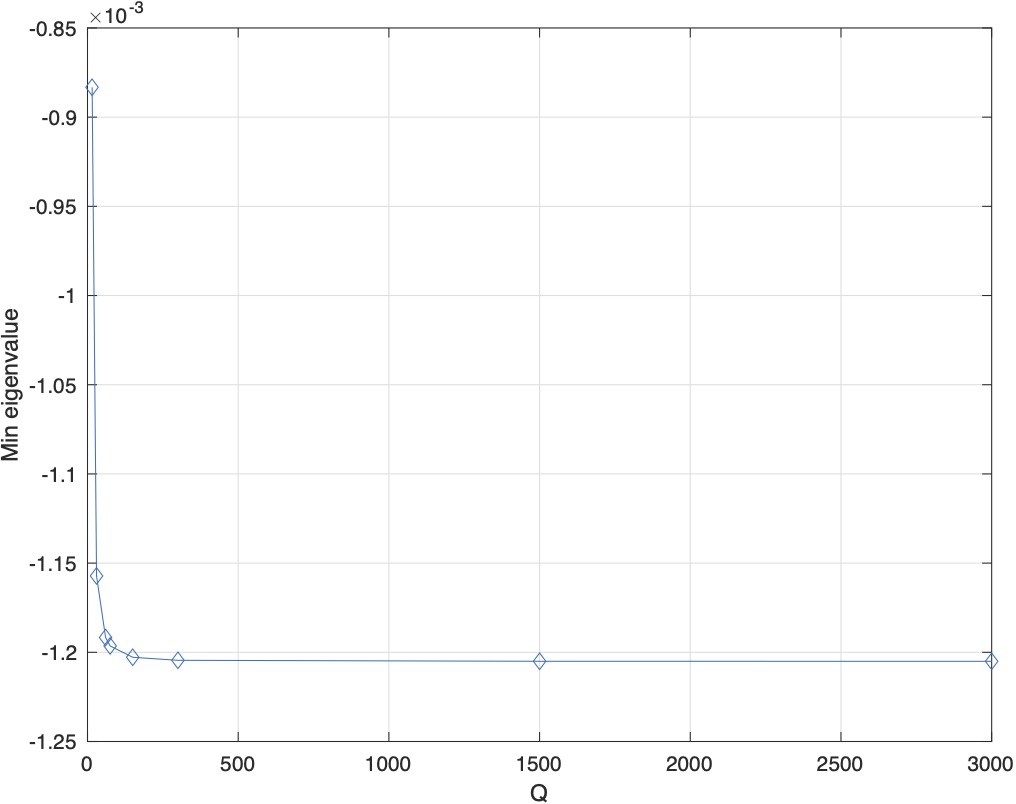}
\caption{The value of the minimum eigenvalue as a function of $Q$.  }
\label{eiewaarde}    
\end{figure}


 \subsection{$Z$ domain and time domain factorization are equivalent if $Q \rightarrow \infty$} \label{formal}

If $Q \rightarrow \infty$ symmetry point equilibrium will clearly hold,  and $\mathbf C$ is locally Toeplitz near the  symmetry point. The matrix $\mathbf G_{\mathrm{adj}}$ is factorized  as 
\begin{equation}
\mathbf G_{\mathrm{adj}} = \mathbf C^{\mathrm T} \mathbf C
\end{equation}
so that for any row near the symmetry point of $\mathbf C^{\mathrm T}$, say row $j$, multiplied with columns $j,j+1,j+2,\cdots$ of $\mathbf C$ yields  a system of non-linear equations given by 
\begin{eqnarray} \label{nonlinear_system_again}
\nonumber c_0^2 + c_1^2 + \cdots + c_{M-1}^2 &=& g_{\mathrm{adj}}\left [{M-1}\right ] \\ 
\nonumber c_0 c_1 + c_1 c_2 + \cdots + c_{M-2} c_{M-1} &=& g_{\mathrm{adj}}\left [{M-2}\right ] \\
c_0 c_2 + c_1 c_3 + \cdots + c_{M-3} c_{M-1} &=& g_{\mathrm{adj}}\left [{M-3}\right ] \\
\nonumber \vdots  &=&  \vdots \\
\nonumber c_0 c_{M-1} &=& g_{\mathrm{adj}}[0].
\end{eqnarray}
This system of non-linear equations is identical to (\ref{nonlinear_system2}), hence in the limit where $Q$ is infinite, time domain factorization is identical to spectral factorization.  It follows that the residual error $\mathbf e$ is given by 
\begin{equation}
    \mathbf e = \mathbf C^T \mathbf C - \mathbf G_{\mathrm{adj}}.
\end{equation}   
As Cholesky decomposition is computed on a digital computer to machine resolution, this result shows that the $\mathrm{L_2}$ norm of $\mathbf e$ will be limited to a lower limit consistent with machine precision.

In Appendix A it is proved that  $\mathbf F$ is unitary, and an all pass filter.  From equation (\ref{transf}) it is known that \begin{equation} 
\mathbf F    \mathbf h_{\mathrm{aug}} = \mathbf C~ \mathbf \delta_\mathrm{aug} \triangleq \mathbf c_{\mathrm{aug}}.
\end{equation}
$\mathbf F$ is an all pass filter, hence it follows that 
 \begin{equation}
|{C(\Omega)}| = |{H(\Omega)}|.	
 \end{equation}
In Section \ref{transform} the salient properties of the matrix $\mathbf F$ will be further explored. 

The results obtained can be summarized as follows.  Denote the stopband negative peak value of $G(\Omega)$ as $\Gamma$ and let $Q \rightarrow \infty$, then it follows that:
\begin{enumerate}
    \item Discrete time and frequency domain factorization are equivalent. This is a known result, see for example \cite{cioffi}. 
    \item The smallest eigenvalue $\lambda_{\mathrm{min}}$ of the Gramian matrix $\mathbf G$ equals  $\Gamma$.  
    \item If $\mathbf g_{\mathrm{adj}}$ is positive definite it can be factored, that is, (\ref{nonlinear_system2}) has a real solution with $E_{L_2} = 0$  if finite computer resolution is neglected.  
    \item A minimum phase filter $\mathbf c$ that follows from a factorization of a positive definite $\mathbf g_{\mathrm{adj}}$ is unique\footnote{As Cholesky decomposition of a positive definite Hermitian matrix is unique.} and optimal --- in the sense that no other real filter exists that will solve (\ref{nonlinear_system2}) with $E_{L_2}$ equal to machine resolution, and has $|{C(\Omega)}| = |{H(\Omega)}|$.
\end{enumerate}

The next section provides experimental  results to confirm the results derived in this section. 

\section{Numerical results: Optimal Chebyshev minimum phase FIR filter design}  \label{FIR_design} 

This section presents a detailed minimum phase FIR filter design based on the results presented in the previous sections. The performance of the filter will be compared to results provided in \cite{antonio}, for a filter that was designed to identical specifications. The difference between the design presented in this section and the design based on best practice \cite{antonio}, comes down to the design of the adjusted linear phase filter $\mathbf g_{\mathrm{adj}}$.   

In this paper a positive definite filter $\mathbf g_{\mathrm{adj}}$ will be the basis of the minimum phase filter design, and the residual error norm $E_{L_2}$ will be limited only by the 64 bit machine resolution.  It will be shown that $E_{L_2}$ is several orders of magnitude smaller than the residual error reported in \cite{antonio}, which is based on best practice design. This section will also demonstrate that $E_{L_2}$ plotted as a function of $\gamma$ exhibits a \emph{waterfall point}, where the filter $\mathbf g_{\mathrm{adj}}$ becomes positive definite.  This is indeed how the appropriate value for $\gamma$ can be determined for high order filters.

It should be noted that the designer must choose a value of $\gamma$ as close to the waterfall point as is possible, as this will prevent leakage in the stopband.  That is, the zeros of the filter $\mathbf c$ are as close to the unit circle as is possible, but with $\mathbf g_{\mathrm{adj}}$ positive definite and factorable as defined in Definition \ref{factorable}.

\subsection{Specifications for a Chebyshev minimum phase filter, and computing the linear phase filter $\mathbf g$}

The specifications for a lowpass minimum phase Chebyshev approximation are the starting point of the design, and defined as follows:
\begin{enumerate}
\item The desired passband ripple value $\delta_p $.  
\item The stopband ripple value $\delta_a$.  
\item The passband edge frequency as $\Omega_p$, in the range $[0,\pi]$.  
\item The stopband edge frequency  $\Omega_s$,  in the range $[0,\pi]$. 
\end{enumerate}
\begin{table}[t]
\scriptsize
\caption{Tap settings for $\mathbf g$} 
\centering 
\begin{tabular}{c c c c} 
\hline\hline 
Tap number $n$ & $\mathbf g$\\ [0.5ex] 
\hline 
$0$ & -0.00033409853951949 \\ 
1 & -0.002489549410806 \\
2 & -0.007656350824928  \\
3 & -0.011354989160955 \\
4 & -0.002981767473881 \\ 
5 &  0.018180581093311 \\
6 & 0.026333770707396  \\
7 & -0.008295888670961 \\
8 & -0.062043244763120 \\
9 & -0.047371546549295 \\
10&  0.095349066618093  \\
11 & 0.295504051520742 \\
12 & 0.391016383693520 \\
13 & 0.295504051520742 \\
14 & 0.095349066618093 \\
15 & -0.047371546549295 \\
16 & -0.062043244763120 \\
17 & -0.008295888670961 \\
18 & 0.026333770707396 \\
19 & 0.018180581093311  \\
20 & -0.002981767473881 \\
21 & -0.011354989160955\\
22 & -0.007656350824928\\
23 & -0.002489549410806\\
24 & -0.00033409853951949\\[1ex] 
\hline 
\end{tabular}
\label{tab1} 
\end{table}

The equivalent ripple parameters for the linear phase filter $\mathbf g$ are given by \cite{antonio}
\begin{eqnarray}
 \Delta_p &=& \frac{4 \delta_p}{2+2 \delta_p^2 -\delta_a^2}, ~~\mathrm{the~passband~ripple} \nonumber \\
 \Delta_a &=& \frac{\delta_a^2}{2+2 \delta_p^2-\delta_a^2}, ~~\mathrm{the~stopband~ripple}.
\end{eqnarray}
A linear phase Chebyshev filter $\mathbf g$  can be designed based on $\Delta_p, \Delta_a, \Omega_p$ and $\Omega_s$, which in practice is performed through an optimal Parks-McClellan algorithm \cite{parks}.   

As an example, tap settings for $\mathbf g$ obtained  with $\Omega_p = 0.3 \pi$, $\Omega_s = 0.6 \pi$, $\delta_p = 0.01$ and $\delta_s = 0.00316$ (these specifications were taken from example $1$ in \cite{antonio}) are shown in Table  \ref{tab1}. 

The frequency response and the pole/zero representation of the filter are shown in Figure \ref{g}.  It is clear that $\mathbf g$ cannot be factored --- and to confirm this, note that the minimum eigenvalue of the Gramian matrix is negative.  
\begin{figure} []
\centering
  \includegraphics[width=0.6\textwidth]{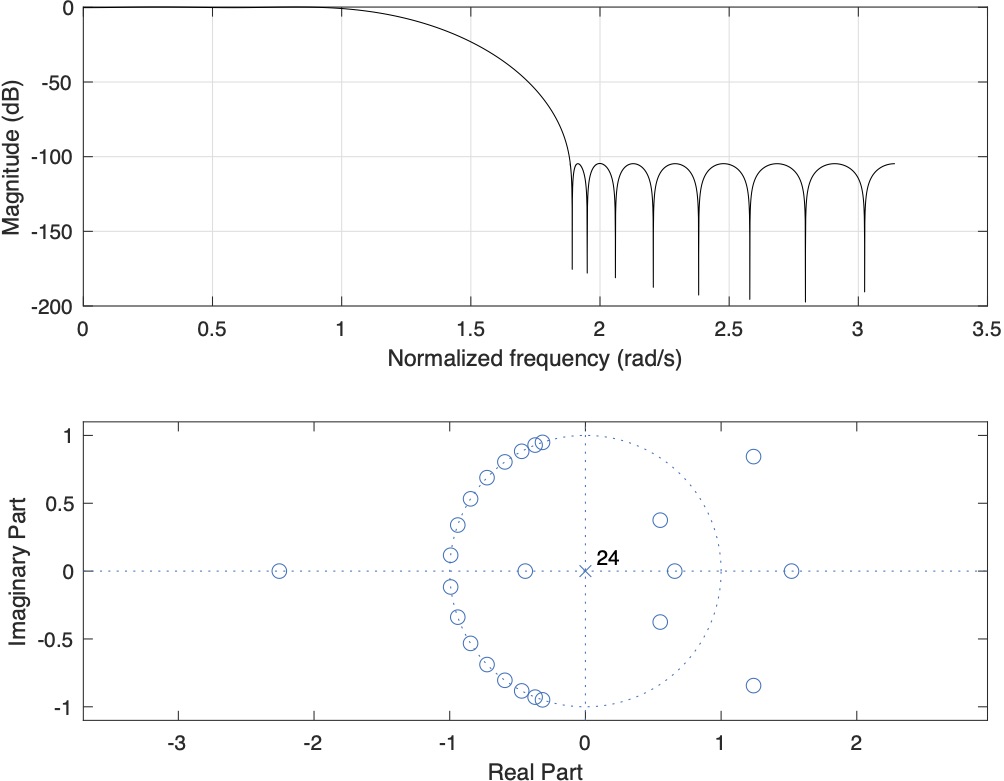}
\caption{A lowpass filter $\mathbf g$ to be factorized in order to compute the minimum phase filter $\mathbf c$.}
\label{g}    
\end{figure}

\subsection{Computing $\mathbf g_{\mathrm{adj}}$ that is factorable}

The next step is to mitigate the non-factorability of the linear phase filter $\mathbf g$ through setting $\gamma$ to an appropriate value and obtaining $\mathbf g_{\mathrm{adj}}$ that is factorable. 

In this case $\gamma_{\mathrm{psd}}$ for a positive semi-definite linear phase filter  can be measured by examining  the negative ripple peaks in $G(\Omega)$. The most negative peak yielded  
\begin{equation}
    \gamma_{\mathrm{psd}} = 5.832240436431935 \times 10^{-6}.
\end{equation}
The residual error $E_{L_2}$ can be computed by solving the non-linear equations (\ref{nonlinear_system2}) through computer optimization.  In this paper a $64$ bit version of the Levenberg-Marquardt optimizer available in MATLAB as \emph{lsqnonlin} were deployed, including the Hessian.  The initial guess for $\mathbf c$ was computed on the time domain using Cholesky decomposition, with $Q=10N$ where $N=25$ in this case (number of taps in $\mathbf g$). 

The residual error $E_{L_2}$ as a function of $\gamma$ is shown in Figure \ref{waterfall}, and the existence of a waterfall point where the error falls away is evident --- this is where the value  $\gamma_{\mathrm{psd}}$ (which is indicated in the figure as the "Ripple" value) is exceeded by $\epsilon \approx 10^{-13}$, rendering $\mathbf g_{\mathrm{adj}}$ factorable and positive definite. After the waterfall point, the residual error is limited by the digital machine resolution (and, as mentioned above, increasing $\gamma$ beyond this value would simply increase the stop-band leakage).  
\begin{figure} []
\centering
  \includegraphics[width=0.75\textwidth]{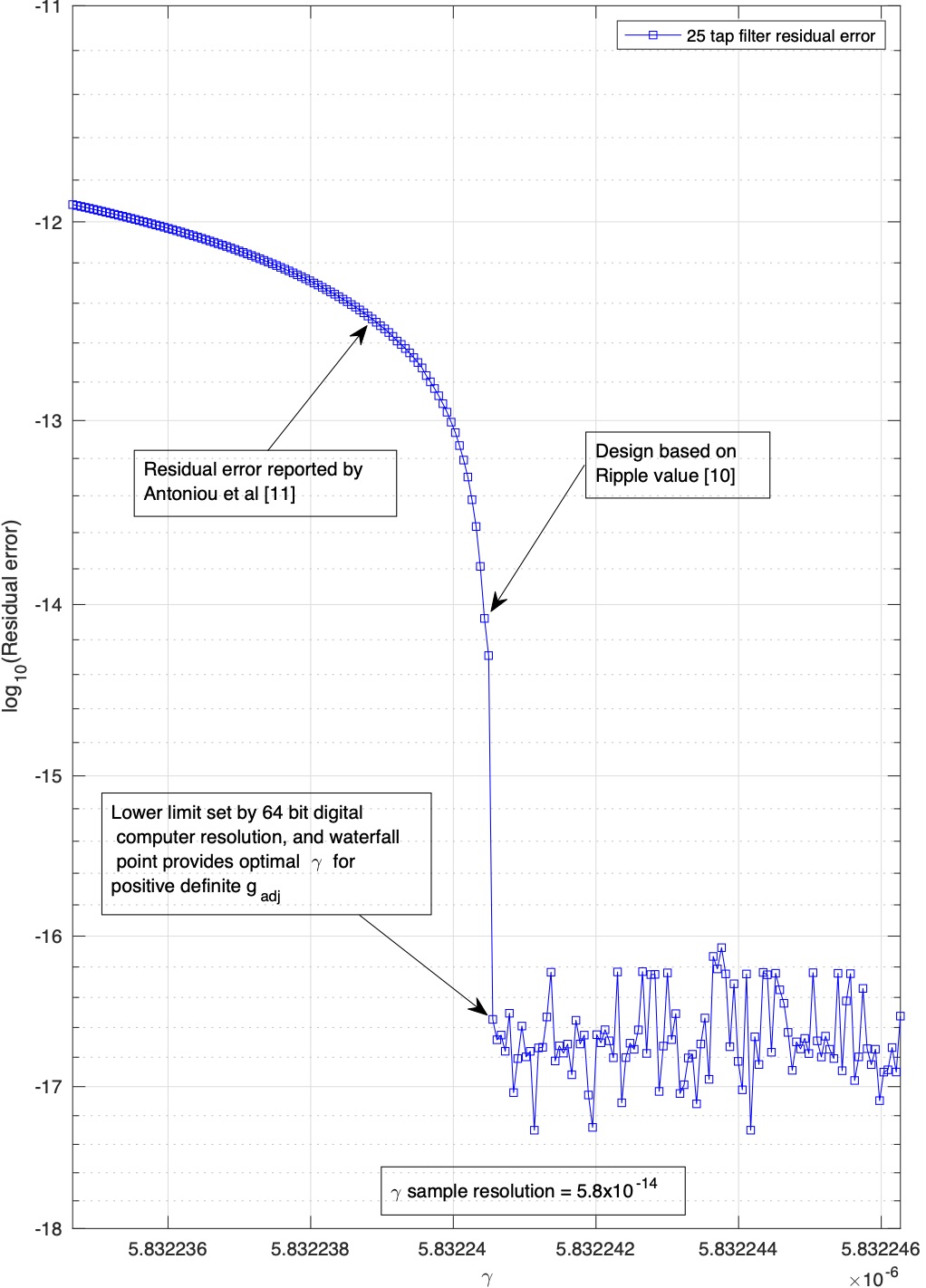}
\caption{The residual error $E_{L_2}$ as a function of $\gamma$.  Note the waterfall point, which occurs as soon as $\mathbf g_{\mathrm{aug}}$ becomes positive definite.}
\label{waterfall}
\end{figure}
Hence $\mathbf g_{\mathrm{aug}}$ can be computed by adding $\gamma$ to the dominant tap of $\mathbf g$, where $\gamma$ is given by 
\begin{equation}
    \gamma = \gamma_{\mathrm{psd}} + 1.16\times 10^{-13}. 
\end{equation}
Also note that the residual error reported by \cite{antonio} (see Figure \ref{waterfall}) is $4$ orders of magnitude above the $64$ bit machine resolution lower limit of $2.8 \times 10^{-17}$. It is unclear how $\gamma$ was measured in \cite{antonio}, but clearly it was short of the positive semi-definite value $\gamma_{\mathrm{psd}}$. 

Note that $\mathbf g_{\mathrm{aug}}$ must be scaled as lifting causes the entire spectrum to lift \cite{antonio}. 

\subsection{Computing the optimum minimum phase filter $\mathbf c$}

 With the optimal value of $\gamma$ and thus a positive definite and factorable $\mathbf g_{\mathrm{aug}}$ obtained as presented above, the optimum minimum phase filter $\mathbf c$ follows by solving the non-linear equations (\ref{nonlinear_system2}) through Levenberg-Marquardt (L-M) optimization.  The approximate taps based on Cholesky decomposition (deployed as an initial guess for the L-M optimization), as well as the final minimum phase filter taps are shown in Table \ref{tab2}, and the frequency response is shown in Figure \ref{optimized}.   
  \begin{figure} []
\centering
  \includegraphics[width=0.55\textwidth]{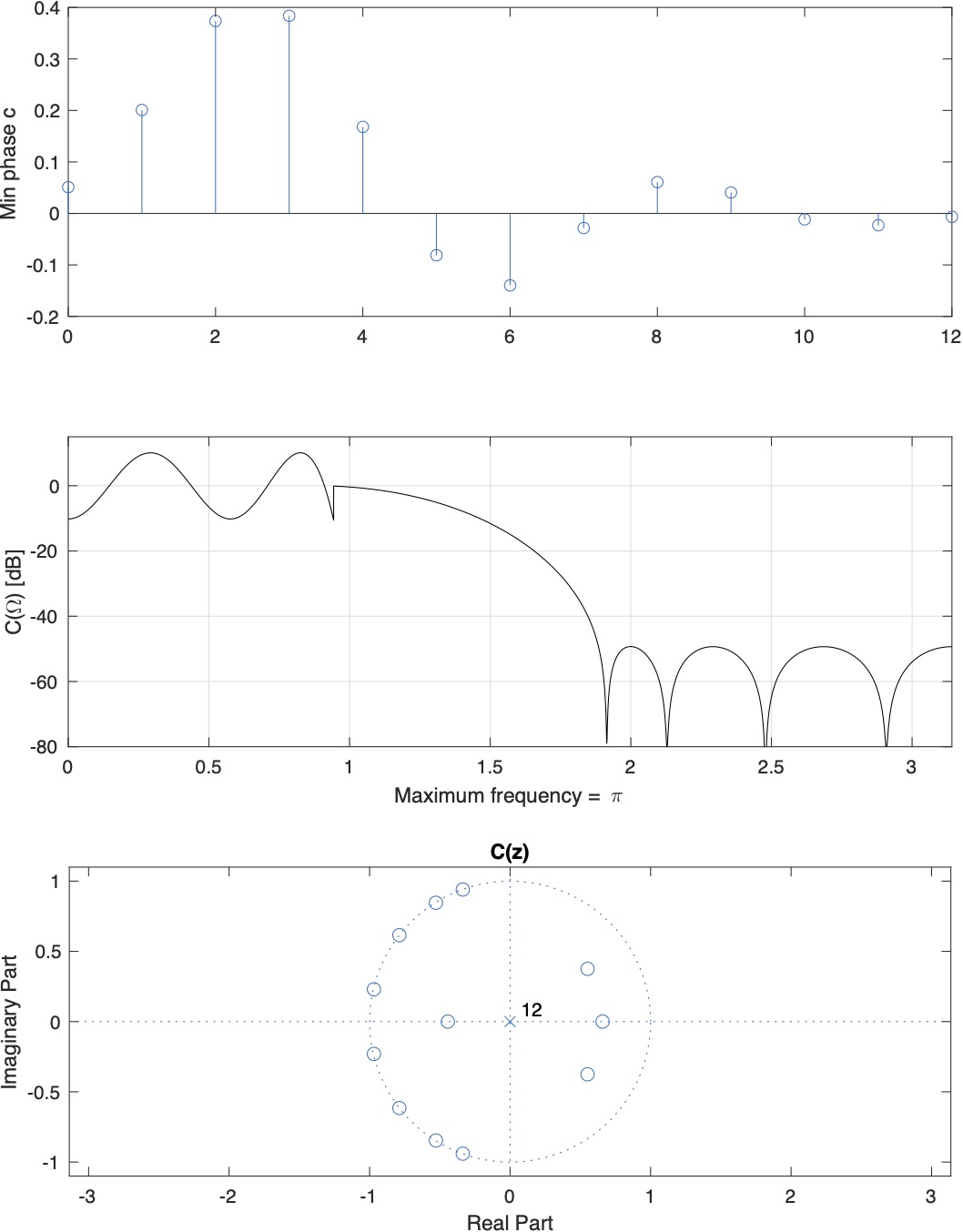}
\caption{The minimum phase FIR  $\mathbf c$  obtained through optimization, based on a positive definite minimum phase filter $\mathbf g_{\mathrm{aug}}$. Note that the magnitude response in the passband was multiplied by a factor of $100$ to aid visualization. }
\label{optimized}    
\end{figure}

The norm of the residual error is shown in Table \ref{L2} for a number of methods available in the literature (as reported in \cite{antonio}), and compared to the error for the proposed design based on a positive definite $\mathbf g_{\mathrm{adj}}$. The filter taps $\mathbf c$ proposed in \cite{antonio} differ in the fifth decimal place when compared to the optimal values shown in Table \ref{tab2}. In practice changes in the fifth decimal place of the taps of a FIR filter would be deemed significant, and thus a positive definite linear phase FIR filter design is an important objective for optimal minimum phase Chebyshev filters.  
\begin{table}[]
\scriptsize
\caption{Settings for the approximate $\mathbf c_{\mathrm{approx}}$ and the  optimal tap settings for $\mathbf c$.} 
\centering 
\begin{tabular}{c c c c} 
\hline\hline 
Tap number $n$ & $\mathbf c$ & $\mathbf c_{\mathrm{approx}} $ ($Q=10 \,N$) \\ [0.5ex] 
\hline 
0 & 0.051115654818476 &0.053170658603589 \\ 
1 & 0.200710190492617  & 0.206160489181734\\
2 & 0.373675662335455  & 0.378228787870903 \\
3 & 0.383763292854902 & 0.380981033331147\\
4 & 0.168110995957186 & 0.159498759291175\\ 
5 & -0.081211240469989 &  -0.086362886469131\\
6 & -0.139792118109026 & -0.137267954224048\\
7 & -0.028413408959431 & -0.024127759326649\\
8 & 0.060844647358445   & 0.061160436307306 \\
9 &   0.040660368538867  & 0.038606283413641 \\
10& -0.011537810825624 & -0.012290539783330 \\
11 & -0.023042542956496 & -0.022479612134730\\
12 & -0.006536978978309  & -0.006283513281458\\[1ex] 
\hline 
\end{tabular}
\label{tab2} 
\end{table}

 \begin{table}[]
\scriptsize
\caption{Residual error norm for various methods and the proposed positive definite design} 
\centering 
\begin{tabular}{c c c c c c } 
\hline\hline 
Proposed &  Factorization \cite{antonio} & Cepstrum & DHT & Orchard \cite{o-w} \\ [0.5ex] \hline 
 $2.84549 \times 10^{-17}$ &  $3.411 \times  10^{-13}$ & $2.394 \times 10^{-6}  $ & $ 2.826 \times 10^{-10} $  & $ 3.376 \times 10^{-10}  $   \\[1ex] 
\hline 
\end{tabular}
\label{L2} 
\end{table}

\section{Transforming an arbitrary phase system to a minimum phase system} \label{transform}

There are applications in practice where a FIR with arbitrary phase is known, but  required to be transformed so that it has a  minimum phase \cite{optics}.   In this section factorization is applied to perform the minimum phase transformation, and results are compared to a transformation based on the MMSE estimator.  The MMSE estimator is widely deployed in practice to perform this transformation  \cite{cioffi}. The system is assumed to operate in additive white Gaussian noise, thus a noise whitening filter need not be deployed.   
 
 \subsection{MMSE estimator deployed to perform the minimum phase transformation}
 
 A model valid for any system with a measured (or given) FIR $\mathbf h$ is given by 
 \begin{equation}
    y[n] = \sum_{k=0}^{M} h[k] s[n-k] + w_s[n]. 
 \end{equation}
 $\mathbf y$ represents an observed but noisy version of the system output in response to the system input $\mathbf s$.  The noise samples $w_s[n]$ are assumed to be drawn from an independently and identically distributed Gaussian random process with variance $\sigma^2$ and zero mean.  
 
 The transformation of $\mathbf h$ to a minimum phase $\mathbf c$ can be performed by a feedforward (anti-causal) filter $\mathbf f$.  Saltz \cite{saltz} proved that such a transformation is possible for the continuous and sampled domain \cite{cioffi}.  That is, the transformed system model is given by 
 \begin{eqnarray}
  \sum_{j=0}^{P} f[j] ~ y[n+j] = \sum_{k=0}^{M} c[k] ~ s[n-k] + w_f[n].
 \end{eqnarray}
The filter $\mathbf f$ has a number of taps $P+1$ that must be quite large relative to $M$.  The MMSE estimator computes filters $\mathbf f$ and $\mathbf c$ based on the orthogonality principle \cite{cioffi}. These ideas are rooted in the seminal work of Wiener and Kolmogorov, and the literature on this methodology is mature and complete \cite{poor,linear_est}.  In this section the MMSE estimator is applied to compute $\mathbf c$, and used as a benchmark to assess factorization applied to the transformation problem.

 \subsection{Numerical results}
 
 The example chosen  is a  random $10$ tap FIR, with a FIR shown in Figure \ref{cioffi_result} (top left).
\begin{figure} []
\centering
  \includegraphics[width=0.9\textwidth]{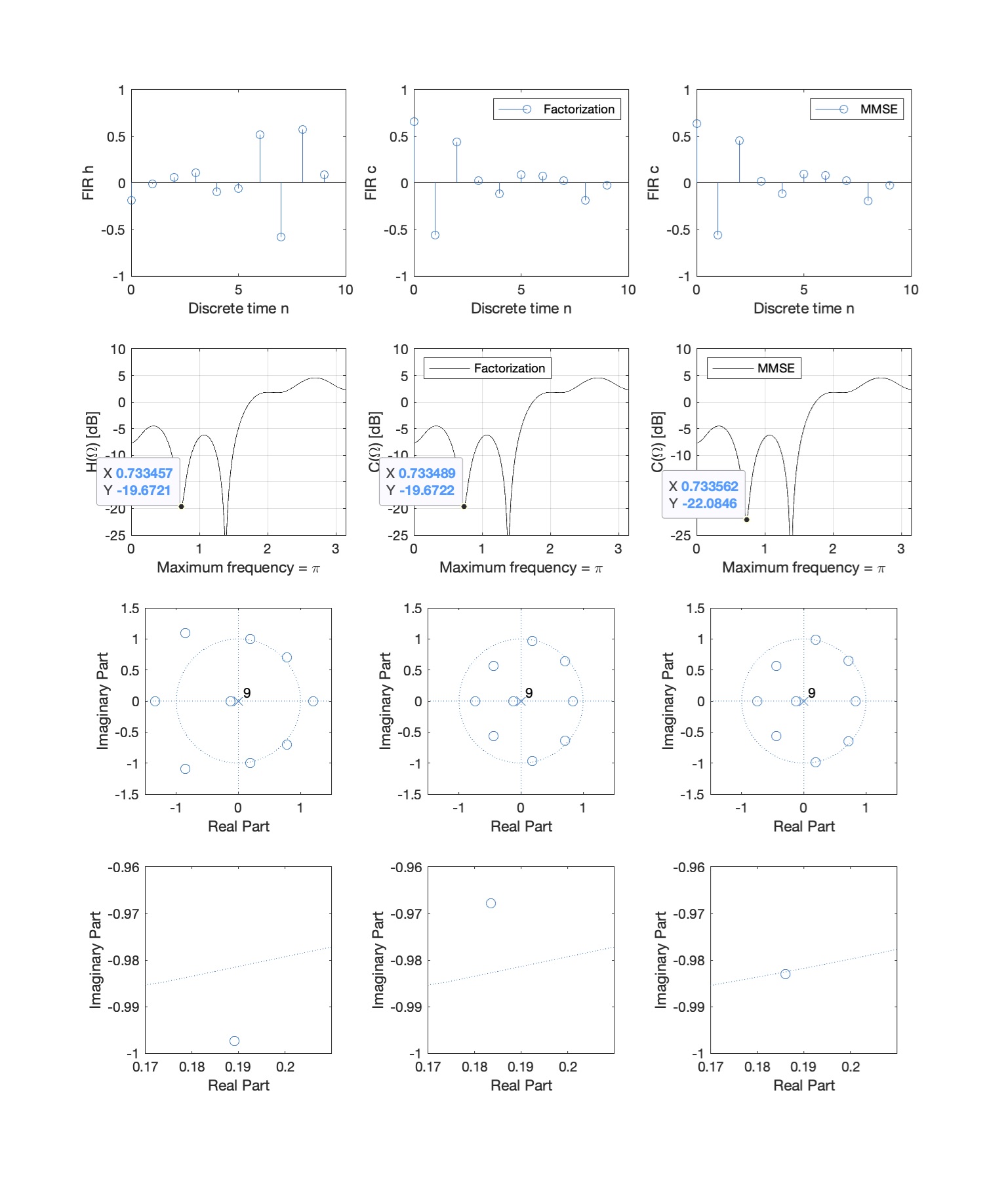}
\caption{A $10$ tap arbitrary phase FIR $\mathbf h$, and the minimum phase FIR $\mathbf{c}$ based on factorization and the MMSE method \cite{cioffi}.}
\label{cioffi_result}    
\end{figure}
Figure \ref{cioffi_result} shows the results for transformation based on factorization, as well as based on  the MMSE estimator. Note that the matched filter $\mathbf g$ on which factorization is based, is a linear phase FIR filter as shown in Figure \ref{cioffi2}, and is factorable.  Hence there is no need for lifting and thus $\gamma = 0$.  

\begin{figure} []
\centering
  \includegraphics[width=0.6\textwidth]{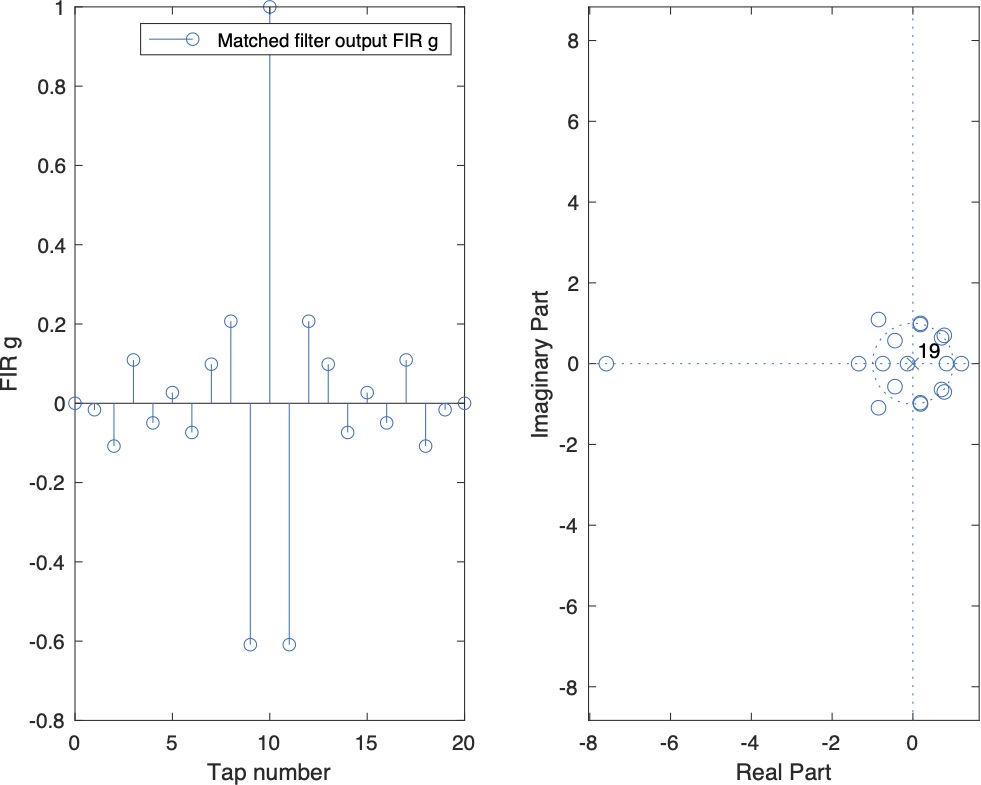}
\caption{The output of the matched filter denoted $\mathbf g$, which is a linear phase filter.}
\label{cioffi2}    
\end{figure}
The MMSE estimator yielded the filter  $\mathbf c$ as shown on the right in Figure \ref{cioffi_result}.   The frequency magnitude plot for the filter based on MMSE shows that there are frequencies where the magnitude response is not identical to that of $H(\Omega)$.  For the proposed transformation the magnitude response is close to $H(\Omega)$ as expected.  

What is concerning is that the MMSE estimator yielded a filter $\mathbf c$ that has a pair of zeros marginally outside the unit circle. In practice where feedback may be applied this would mean that the filter $\mathbf c$ is not unconditionally stable.  In contrast factorization yields the correct placement for this zero pair, as shown in Figure \ref{cioffi_result}.

The MMSE estimator does not provide a residual error that is anywhere near the machine resolution floor, as shown in Table \ref{mmse_table}.  Evidently an error on the order of $10^{-3}$ provides a filter $\mathbf c$ where $C(\Omega)$ is visibly sub-optimal as shown in Figure \ref{cioffi_result}. 
\begin{table}[]
\scriptsize
\caption{Residual error norm for MMSE and the proposed method} 
\centering 
\begin{tabular}{c c c c c c } 
\hline\hline 
Proposed & MMSE \cite{cioffi}  \\ [0.5ex] \hline 
 $6.11 \times 10^{-17}$ &  $2.07 \times  10^{-3}$    \\[1ex] 
\hline 
\end{tabular}
\label{mmse_table} 
\end{table}

It is interesting to consider how factorization is able to obtain these results. Factorization yielded small but critical changes in the filter $\mathbf c$ when compared to MMSE, and these changes were able to guarantee that $\mathbf c$ is minimum phase.  The difference comes down to the feedforward (anti-casual) filter $\mathbf f$, as shown in Figure \ref{prefilter}.  For factorization, the feedforward filter is represented by the matrix $\mathbf F$, which was demonstrated to be unitary and an all pass filter.  Factorization computes the  matrix $\mathbf F$ based on Cholesky decomposition, not based on the orthogonality principle as is the case for the MMSE estimator.  The rows of matrix $\mathbf F$ (near the symmetry point where it is locally Toeplitz) contain $\mathbf f$, and exhibits enhanced accuracy due to the Cholesky decomposition. 

Finally, note that matrix $\mathbf F$ is lower triangular\footnote{Because the definition of the direction of time $n$ was reversed to conform to the upper triangular $\mathbf C$, a lower triangular matrix is anti-causal in this paper. Most Cholesky decomposition functions in commercial software yield an upper triangular matrix, for historical reasons.}, confirming the results in \cite{saltz}, that is, the transformation filter is purely anti-causal. 

\begin{figure} [h]
\centering
  \includegraphics[width=0.6\textwidth]{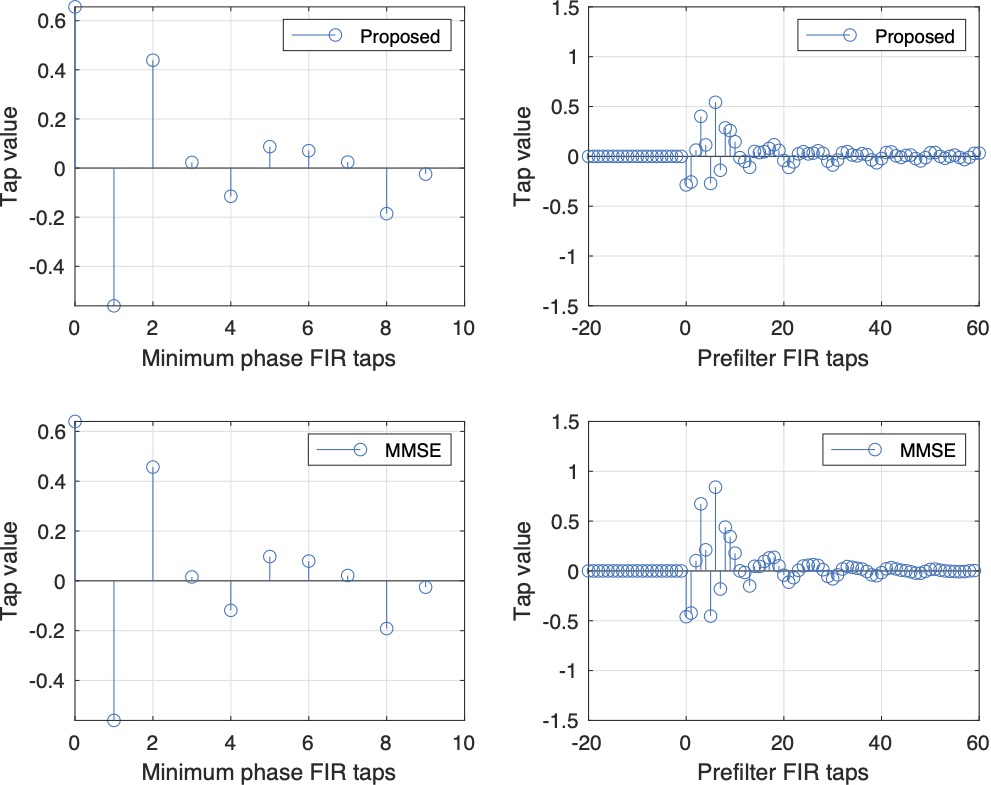}
\caption{The anti-causal transformation filter $\mathbf f$ able to transform the given FIR $\mathbf h$ to a minimum phase filter $\mathbf c$.}
\label{prefilter}    
\end{figure}

\section{Conclusions} \label{conclude}

The paper demonstrated that if the adjustment (or lifting) factor $\gamma$ is set to the ripple peak, as proposed in \cite{herrmann}, it does not produce a factorable equivalent for a linear-phase Chebyshev filter, and is in practice insufficient for that goal. The paper demonstrated that the definition of what constitutes a factorable linear phase filter has to be modified:  it was proposed that the non-linear equations proposed in  \cite{o-w} be exactly solvable (up to machine precision error) by a \emph{real} minimum phase filter, before a linear phase filter is deemed factorable.  

The paper then showed that this definition of factorability demands that the linear phase filter be positive definite, not positive semi-definite as proposed by \cite{herrmann}. 
By performing a time-domain factorization of the FIR system, it was shown that the change in the definition of a factorable linear phase filter is correct: the time domain was  able to pinpoint the exact adjustment required to ensure a solvable system --- the Gramian represents the linear phase filter and must be positive definite. 

The consequences of a change from positive semi-definite to positive definite is dramatic, and when the residual error is plotted as a function of the lifting factor $\gamma$, it exhibits a waterfall point coinciding with the point where $\gamma$ guarantees factorability. At that point the residual error falls away to an error floor set by the finite resolution of the digital computer.  

Time-domain factorization was shown to be possible if an augmentation of the impulse response is defined, and becomes exact as the expansion approaches infinity. However, even for relatively small expansion factors useful approximations for the adjustment value and the filter coefficients are obtained. This served to provide a good initial starting point for optimization of the non-linear equations proposed in \cite{o-w}. 

This process is useful for any minimum-phase FIR filter design and/or transformation, and the approximate time domain solution for the filter coefficients produces rapid and accurate convergence of subsequent iterative methods. The improvements in accuracy beyond those achieved in \cite{antonio} were reflected in changes to the fifth decimal place of the minimum phase filter taps.   But apart from this improvement which may not be  of practical importance for Chebyshev filter design (depending on the application), the proposed methods also have the benefit of generality and theoretical clarity.

\bibliography{mybibfile}

\begin{thebibliography}{10}
\expandafter\ifx\csname url\endcsname\relax
  \def\url#1{\texttt{#1}}\fi
\expandafter\ifx\csname urlprefix\endcsname\relax\def\urlprefix{URL }\fi
\expandafter\ifx\csname href\endcsname\relax
  \def\href#1#2{#2} \def\path#1{#1}\fi

\bibitem{oppenheim}
A.~V. Oppenheim, R.~W. Schafer, Digital Signal Processing, Pearson Education,
  2015.

\bibitem{oppenheim2}
A.~V. Oppenheim, R.~W. Schafer, Discrete-Time Signal Processing, Prentice-Hall,
  1989.

\bibitem{smith}
Smith, \href{https://ccrma.stanford.edu/jos/}{Introduction to Digital Filters
  with Audio Applications}.
\newline\urlprefix\url{https://ccrma.stanford.edu/jos/}

\bibitem{proakis2}
J.~G. Proakis, D.~G. Manolakis, Digital signal processing: Principles,
  algorithms, and applications, New York: Macmillan, 1992.

\bibitem{linear_syst}
T.~Kailath, Linear Systems, Prentice-Hall, 1980.

\bibitem{parks}
L.~R. Rabiner, J.~H. McClellan, T.~W. Parks, Fir digital filter design
  techniques using weighted chebyshev approximation, Proceedings of the IEEE 63
  (1975) 595--610.

\bibitem{evans}
N.~D. Venkata, B.~L. Evans, S.~R. McCaslin, Design of optimal minimum-phase
  digital fir filters using discrete hilbert transforms, IEEE Trans. Signal
  Process. 48 (2000) 1491--1495.

\bibitem{kamp}
Y.~Kamp, C.~J. Wellekens, Optimal design of minimum-phase fir filters, IEEE
  Trans. Acoust., Speech, Signal Process. 31 (1983) 922--926.

\bibitem{reddy}
G.~R. Reddy, Design of minimum-phase fir digital filter through cepstrum,
  Electron. Lett. 22 (1986) 1225--1227.

\bibitem{herrmann}
O.~Herrmann, H.~W. Schuessler, Design of non-recursive digital filters with
  minimum phase, Electron. Lett. 6 (1970) 329--330.

\bibitem{antonio}
S.~Kidambi, A.~Antoniou, Design of minimum-phase filters using optimization,
  IEEE Trans. Circuits and Systems -- II: Express briefs 64 (2017) 472--476.

\bibitem{o-w}
H.~J. Orchard, A.~N. Wilson, On the computation of a minimum phase spectral
  factor, IEEE Trans. Circuits Syst. I, Fundam. Theory Appl. 50 (2003)
  365--375.

\bibitem{proakis1}
J.~G. Proakis, M.~Salehi, Digital Communications, Boston: McGraw-Hill, 2008.

\bibitem{optics}
R.~Halir, I.~Molina-Fernández, J.~Wangüemert-Pérez, A.~Ortega-Moñux,
  J.~de~Oliva-Rubio, P.~Cheben, Characterization of integrated photonic devices
  with minimum phase technique, Optics express 17 (2009) 8349--61.
\newblock \href {https://doi.org/10.1364/OE.17.008349}
  {\path{doi:10.1364/OE.17.008349}}.

\bibitem{cioffi}
N.~Al-Dhahir, J.~Cioffi, {MMSE} decision-feedback equalizers: finite-length
  results, IEEE Transactions on Information Theory 41~(4) (1995) 961--975.
\newblock \href {https://doi.org/10.1109/18.391242}
  {\path{doi:10.1109/18.391242}}.

\bibitem{linear_alg}
G.~H. Golub, C.~F. {Van Loan}, Matrix Computations, 2013.

\bibitem{regular}
Regularization,
  \url{https://math.stackexchange.com/questions/3596910/regularization-of-a-matrix-using-a-diagonal-matrix},
  accessed: 2021-09-21.

\bibitem{saltz}
J.~Salz, Optimum mean-square decision feedback equalization, The Bell System
  Technical Journal 52~(8) (1973) 1341--1373.
\newblock \href {https://doi.org/10.1002/j.1538-7305.1973.tb02023.x}
  {\path{doi:10.1002/j.1538-7305.1973.tb02023.x}}.

\bibitem{poor}
K.~{Vastola}, H.~{Poor}, Robust wiener-kolmogorov theory, IEEE Transactions on
  Information Theory 30~(2) (1984) 316--327.
\newblock \href {https://doi.org/10.1109/TIT.1984.1056875}
  {\path{doi:10.1109/TIT.1984.1056875}}.

\bibitem{linear_est}
T.~Kailath, A.~H. Sayed, B.~Hassibi, Linear Estimation, Upper Saddle River,
  N.J., Prentice Hall, 2000.

\end{thebibliography}

\newpage

\section*{Appendix A} \label{appendix1}

\newtheorem{thm}{Theorem}

\begin{thm} Transformation matrix $\mathbf F$ is unitary \end{thm}

\begin{proof}
The proof starts by  computing 
\begin{equation}
\mathbf{F F^{\mathrm T} = \left ( \left ( C^{\mathrm T} \right ) ^{-1}  H^{\mathrm T} \right ) \left ( \left ( C^{\mathrm T} \right ) ^{-1} H^{\mathrm T} \right )^{\mathrm T}}
\end{equation}
and then demonstrates that this operation yields an identity matrix $\mathbf I$.   Proceeding in a step by step manner, it follows that 
 \begin{eqnarray} 
\mathbf{ \left ( \left ( C^{\mathrm T} \right ) ^{-1} H^{\mathrm T} \right ) \left ( \left ( C^{\mathrm T} \right ) ^{-1} H^{\mathrm T} \right ) ^{\mathrm T}}  
 \noindent & \nonumber\\ 
 = \mathbf{ \left ( \left ( C^{\mathrm T} \right ) ^{-1} H^{\mathrm T} \right ) \left (  \left ( H^{\mathrm T} \right ) ^{\mathrm T}  \left [ \left ( C^{\mathrm T} \right ) ^{-1} \right ] ^{\mathrm T}     \right )} \nonumber & \\
 = \mathbf{ \left ( \left ( C^{\mathrm T} \right ) ^{-1} H^{\mathrm T} \right ) \left (   H   \left [ \left ( C^{\mathrm T} \right ) ^{-1} \right ] ^{\mathrm T}     \right ) }\nonumber & \\
 = \mathbf{ \left ( \left ( C^{\mathrm T} \right ) ^{-1} H^{\mathrm T} \right ) \left ( H  C  ^{-1}     \right ) }\nonumber & \\ 
 =  \mathbf{  (C^{\mathrm T})^{-1} H^{\mathrm T}     H C^{-1}}. &
 \end{eqnarray}
But $\mathbf C$  was defined through factorization hence 
\begin{equation}
\mathbf{ H^{\mathrm T} H = C^{\mathrm T} C}.
\end{equation}
Substitute this expression into the previous equation and obtain
\begin{eqnarray}
\mathbf{ F F^{\mathrm T} = \left ( \left ( C^{\mathrm T} \right ) ^{-1} H^{\mathrm T} \right ) \left ( \left ( C^{\mathrm T} \right ) ^{-1} H^{\mathrm T} \right ) ^{\mathrm T}} &=& \mathbf{  (C^{\mathrm T})^{-1} H^{\mathrm T}  H    C  ^{-1}} \nonumber \\
&=& \mathbf{  (C^{\mathrm T})^{-1} C^{\mathrm T}   C   C^{-1}} \nonumber \\
&=& \mathbf I.
\end{eqnarray}
  Hence  the transformation  matrix $\mathbf F$ is unitary.
\end{proof}


\begin{thm}  \label{all_pass} {$\mathbf F$ is an all pass filter. } \end{thm}

\begin{proof}
Denote a random white noise sequence as a column vector $\mathbf n$, with a noise covariance matrix given by  $\mathbf I$ the identity matrix.  Write the noise covariance matrix of filtered noise $\mathbf{n_f} = \mathbf{ F n}$  as 
 \begin{equation}
 {\mathbf \Sigma}  = E\{ \mathbf{ n_f n_f}^{\mathrm T} \}.
 \end{equation}
   It follows that
\begin{equation}
 \mathbf{ \Sigma} =  E\{ \mathbf{ F n  (F n)}^{\mathrm T} \}  = E \{  \mathbf{ F n n^\mathrm{T} F^\mathrm{T}}\}  = {\mathbf F}  \, E \{ \mathbf{ n n^\mathrm{T}\} \, F}^{\mathrm T} .  
 \end{equation}  
 But since $\mathbf n$ is white noise, $\mathbf{ E \{ n n^\mathrm{T}\} = I}$, and hence 
 \begin{equation}
 \mathbf{ \Sigma  = F F^\mathrm{T}}. 
 \end{equation}  
 However  $\mathbf F$ is unitary,  thus it follows that   
 \begin{equation}
{\mathbf \Sigma} = E\{ \mathbf{n_f {n_f}^\mathrm{T}\} = F F^\mathrm{T}  = I}. 
 \end{equation}  
Hence the  filtered noise $\mathbf{n_f = F n}$ has a white power spectral density, and $\mathbf F$ is an all-pass filter. 
\end{proof}

\end{document}